\documentclass[prx,aps,superscriptaddress,twocolumn,nofootinbib,longbibliography]{revtex4-2}

\usepackage{graphicx,color,amsmath,amsfonts,enumerate,amsthm,amssymb,mathtools,enumitem,thmtools,hyperref,subfigure,mathdots,enumitem,centernot,bm,soul,bbm,tikz,pgfplots,float}
\usepackage[capitalise, noabbrev]{cleveref}
\usepackage[normalem]{ulem}
\usepackage{tcolorbox}
\usepackage{ragged2e}
\hypersetup{colorlinks=true,linkcolor=blue,citecolor=blue,urlcolor=blue}
\pgfplotsset{compat=newest}
\tcbset{before skip=10pt,toptitle=2mm,bottomtitle=1mm,fonttitle=\bfseries}
\tcbuselibrary{theorems}
\tcbuselibrary{breakable}

\definecolor{NavyBlue}{rgb}{0.0, 0.0, 0.5}
\definecolor{OliveGreen}{rgb}{0.33, 0.42, 0.18}
\definecolor{def_color_frame}{RGB}{220,230,242}
\colorlet{def_color_back}{def_color_frame!30}
\definecolor{def_color_text}{RGB}{37,64,97}
\definecolor{def_color_frame2}{RGB}{242,200,200}
\colorlet{def_color_back2}{def_color_frame2!30}
\definecolor{def_color_text2}{RGB}{97,55,33}

\newtcolorbox[auto counter]{bluebox}[2][]{%
colback=def_color_back,colframe=def_color_frame,fonttitle=\bfseries,coltitle=def_color_text,float,floatplacement=t,title=Box~\thetcbcounter: #2,#1}

\newtcolorbox[use counter from=bluebox]{redbox}[2][]{%
colback=def_color_back2,colframe=def_color_frame2,fonttitle=\bfseries,coltitle=def_color_text2,float,floatplacement=t,title=Box~\thetcbcounter: #2,#1}

\def\D{ {\cal D} }
\def\E{ {\cal E} }
\def\N{ {\cal N} }
\def\F{ {\cal F} }

\def\>{\rangle}
\def\<{\langle}
\newcommand{\bra}[1]{\langle {#1} |}
\newcommand{\ket}[1]{| {#1} \rangle}
 
\newcommand{\ketbra}[2]{\ensuremath{|#1\rangle\!\langle#2|}}
\newcommand{\matrixel}[3]{\ensuremath{\left\langle #1 \vphantom{#2#3} \right| #2 \left| #3 \vphantom{#1#2} \right\rangle}}
\newcommand{\tr}[1]{\mathrm{Tr}\left( #1 \right)}
\newcommand{\trr}[2]{\mathrm{Tr}_{#1}\left( #2 \right)}
\newcommand{\iden}{\mathbbm{1}}

\renewcommand{\v}[1]{\ensuremath{\boldsymbol #1}}

\definecolor{bluecyan}{rgb}{0.27, 0.66, 0.88}
\definecolor{ppblue}{RGB}{46,117,182}
\definecolor{ppred}{RGB}{197, 90, 17}

\theoremstyle{plain}
\newtheorem{thm}{Theorem}
\newtheorem{subthm}{Theorem}[thm]

\newtheorem{lem}[thm]{Lemma}

\theoremstyle{definition}

\begin{document}

\title{Fluctuation-dissipation relations for thermodynamic distillation processes}

\author{Tanmoy Biswas}
\affiliation{International Centre for Theory of Quantum Technologies, University of Gdansk, Wita Stwosza 63, 80-308 Gdansk, Poland.}
%\orcid{0000-0001-7232-976X}

\author{A. de Oliveira Junior}
\affiliation{Faculty of Physics, Astronomy and Applied Computer Science, Jagiellonian University, 30-348 Kraków, Poland.}
%\orcid{0000-0002-5319-0835}

\author{Micha{\l} Horodecki}
\affiliation{International Centre for Theory of Quantum Technologies, University of Gdansk, Wita Stwosza 63, 80-308 Gdansk, Poland.}
%\orcid{0000-0002-0446-3059}

\author{Kamil Korzekwa}
\affiliation{Faculty of Physics, Astronomy and Applied Computer Science, Jagiellonian University, 30-348 Kraków, Poland.}
%\orcid{0000-0002-0683-5469}

\begin{abstract}

The fluctuation-dissipation theorem is a fundamental result in statistical physics that establishes a connection between the response of a system subject to a perturbation and the fluctuations associated with observables in equilibrium. Here we derive its version within a resource-theoretic framework, where one investigates optimal quantum state transitions under thermodynamic constraints. More precisely, we first characterise optimal thermodynamic distillation processes, and then prove a relation between the amount of free energy dissipated in such processes and the free energy fluctuations of the initial state of the system. Our results apply to initial states given by either asymptotically many identical pure systems or arbitrary number of independent energy-incoherent systems, and allow not only for a state transformation, but also for the change of Hamiltonian. The fluctuation-dissipation relations we derive enable us to find the optimal performance of thermodynamic protocols such as work extraction, information erasure and thermodynamically-free communication, up to second-order asymptotics in the number $N$ of processed systems. We thus provide a first rigorous analysis of these thermodynamic protocols for quantum states with coherence between different energy eigenstates in the intermediate regime of large but finite $N$.

\end{abstract}

\maketitle

% ------------------------------------------------
% SECTION I - INTRODUCTION
% ------------------------------------------------

\section{Introduction}
\label{sec:intro}

Thermodynamics has been profoundly triumphant by impacting the natural sciences and allowing the development of technologies that go from coolers to spaceships. As a theory of macroscopic systems in equilibrium, it presents us with a compelling picture of what state transformations are allowed in terms of a small number of macroscopic quantities, such as work and entropy~\cite{caratheodory1909untersuchungen,giles1964mathematical}. The drawback of the macroscopic description is that thermodynamics inevitably deals with average quantities, and as systems get smaller, fluctuations of these quantities become increasingly relevant, requiring a new description~\cite{Seifert2008, Seifert_2012}. Going beyond the original scenario of equilibrium thermodynamics led scientists to investigate fluctuations around these averages and their impact on the system dynamics. This line of research dates back to Einstein and Smoluchowski, who derived the connection between fluctuation and dissipation effects for Brownian particles~\cite{Einstein1906, Smoluchowski1906}. Now, it is well known that near-equilibrium, linear response theory provides a general proof of the fluctuation-dissipation theorem, which states that the response of a given system when subject to an external perturbation is expressed in terms of the fluctuation properties of the system in thermal equilibrium~\cite{Kubo1966,Bonanca2008}. The theoretical description underlying the fluctuation-dissipation relations is usually expressed in terms of the stochastic character of thermodynamic variables. This approach is strongly motivated since it is experimentally viable~\cite{batalhao2014, An2015}. 

On the other hand, a complementary approach is based on resource theories~\cite{Janzing2000,horodecki2013quantumness, brandao2015second,Lostaglio2019}. It aims at going beyond the thermodynamic limit and the assumption of equilibrium, and is often presented as an extension of statistical mechanics to scenarios with large fluctuations, the so-called single-shot statistical mechanics~\cite{Dahlsten_2011,Yunger2018}. A natural question is then whether fluctuation-dissipation relations are present in such a resource-theoretic description. Although important insights have been obtained in trying to connect the information-theoretic and fluctuation theorem approaches~\cite{Alhambra2016, Halpern2015}, they have, so far, not been explicitly related to dissipation. The tools required for the analysis of free energy dissipation in a resource-theoretic framework were developed in Refs.~\cite{Hayashithermo2017, Chubb2018beyondthermodynamic,Chubb2019_2,korzekwa2019avoiding}, where the authors investigated irreversibility of thermodynamic processes due to finite-size effects. However, the relation between fluctuations and actual dissipation was not derived and, moreover, these results were obtained for quasi-classical case of energy-incoherent states. Thus, they were not able to account for quantum effects that come into play when dealing with even smaller systems, when fluctuations around thermodynamic averages are no longer just thermal in their origin.

This work pushes towards a genuinely quantum framework characterising optimal thermodynamic state transformations and links fluctuations with free energy dissipation. We investigate a special case of state interconversion processes known as thermodynamic distillations. These are thermodynamic processes in which a given initial quantum system is transformed, with some transformation error, to a pure energy eigenstate of the final system. In particular, we focus on the initial system consisting of asymptotically many non-interacting subsystems that are either energy-incoherent and non-identical (in different states and with different Hamiltonians), or pure and identical. Within this setting, our main results are given by two theorems. The first one yields the optimal transformation error as a function of the free energy difference between the initial and target states, and the free energy fluctuations in the initial state (see Table~\ref{box1}). This can be seen as an extension of previously derived results on optimal thermodynamic state transformations~\cite{Hayashithermo2017, Chubb2018beyondthermodynamic,Chubb2019_2,korzekwa2019avoiding}. The second theorem provides a precise relation between the free energy fluctuations of the initial state and the minimal amount of free energy dissipated in the optimal thermodynamic distillation process (see Table~\ref{box2}). It is conceptually novel and does not form an extension of previously known results, and as such it constitutes our main contribution. Note that the second theorem employs the first one as one of the building blocks.

% FIRST BOX

\begin{table}[t]
\setlength{\tabcolsep}{7pt}
\centering
\begin{tabular}{|l|}
\hline 
\begin{minipage}{0.45\textwidth}
\vspace{1.5mm}
\textbf{Optimal transformation error for thermodynamic distillation processes} 
\vspace{1.5mm}
\end{minipage}
\\
 \hline
\\
\begin{minipage}{0.45\textwidth}
\justifying
\vbox{\noindent In the optimal thermodynamic process of \mbox{$\epsilon$-approximate} transformation from many independent non-equilibrium systems into systems without fluctuations of free energy, the transformation error $\epsilon$ satisfies}
\vspace{2mm}
\vbox{ \begin{equation} \epsilon = \Phi\left(-\frac{\Delta F}{\sigma(F)} \right), \end{equation}} 
\vspace{2mm}
\vbox{\noindent where $\Delta F$ is the free energy difference between the initial and target state, $\sigma(F)$ is the free energy fluctuation in the initial state, and} 
\vspace{2mm}
\vbox{ \begin{equation} \Phi(x) := \frac{1}{\sqrt{2\pi}}\int_{-\infty}^x e^{-t^2/2}dt .\end{equation}} 
\vspace{2mm}
\vbox{\noindent We have proved such a statement for many independent systems in arbitrary incoherent states, as well as for many independent and identical systems in the same pure state. We conjecture that this is true for many independent systems in arbitrary mixed states.}
\vspace{2mm}
\end{minipage}
\\ \hline
\end{tabular}
\caption{Summary of the first main result.\label{box1}}
\end{table}

% SECOND BOX

\begin{table}[t]
\setlength{\tabcolsep}{7pt}
\centering
\begin{tabular}{|l|}
\hline
\begin{minipage}{0.45\textwidth}
\vspace{1.5mm}
\textbf{Fluctuation-dissipation relation for
thermodynamic distillation processes} 
\vspace{1.5mm}
\end{minipage}
\\
 \hline
\\
\begin{minipage}{0.45\textwidth}
\justifying
\vbox{\noindent In the optimal thermodynamic process of \mbox{$\epsilon$-approximate} transformation from many independent non-equilibrium systems into systems without fluctuations of free energy, the dissipated free energy satisfies }
\vspace{2mm}
\vbox{\begin{equation}\label{eq:diss_intro}  F_{\rm diss}= a(\epsilon)\, \sigma(F), \end{equation}}
\vspace{2mm}
\vbox{\noindent where $\sigma(F)$ is the free energy fluctuation in the initial state, and}
\vbox{\begin{equation} \label{eq:a_intro}
    a(\epsilon)=-\Phi^{-1}(\epsilon)(1-\epsilon)+\frac{\exp\Big(\frac{-(\Phi^{-1}(\epsilon))^2}{2}\Big)}{\sqrt{2\pi}},  \end{equation}} 
\vbox{\noindent with $\Phi^{-1}(x)$ being the inverse of the Gaussian cumulative distribution function.}

\vbox{\begin{center}
\includegraphics{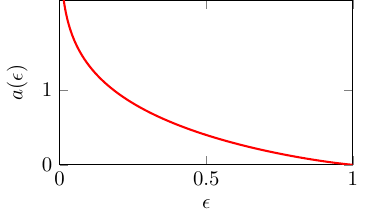} \end{center} }

\vbox{\noindent We have proved this statement for many independent systems in identical incoherent states, while for for many independent systems in identical pure states we proved that the right hand side of Eq.~\eqref{eq:diss_intro} is a lower bound for $F_{\mathrm{diss}}$. We conjecture that the equality in Eq.~\eqref{eq:diss_intro} holds for many independent systems in arbitrary mixed states.} 
\vspace{2mm}
\end{minipage}
\\ \hline
\end{tabular}
\caption{Summary of the second main result.\label{box2}}
\end{table}

Our results allow us for a rigorous study of important thermodynamic protocols. First of all, we extend the analysis of work extraction to the regime of not necessarily identical incoherent states, as well as to pure states. By directly applying our main results, we obtain a second-order asymptotic expression for the optimal transformation error while extracting a given amount of work from the initial system. Moreover, we also verify the accuracy of the obtained expression by comparing it with the numerically optimised work extraction process. As a second application, we analyse the optimal energetic cost of erasing $N$ independent bits prepared in arbitrary states. In this case, we obtained the optimal transformation error for the erasure process as a function of invested work. The last application we consider is the optimal thermodynamically-free communication scheme, i.e., the optimal encoding of information into a quantum system without using any extra thermodynamic resources. Applying our theorems gives us the optimal number of messages that can be encoded into a quantum system in a thermodynamically free-way, which we show to be directly related to the non-equilibrium free energy of the system. This result can be interpreted as the inverse of the Szilard engine, as in this process we use the ability to perform work to encode information. Furthermore, our results connect the fluctuations of free energy and the optimal average decoding error. Finally, our findings also provide new tools to study approximate transformations and corresponding asymptotic interconversion rates. Here, we not only extend previous distillation results~\cite{Chubb2018beyondthermodynamic} to non-identical systems, but also to genuinely quantum states in superposition of different energy eigenstates. 

The paper is organised as follows. We start in Sec.~\ref{sec:high} with a high-level description that can give a flavour of our investigations and explains the physical intuition behind them to a broad audience without the necessity to get into the technicalities of the framework we work in. We then recall the resource-theoretic approach to thermodynamics in Section~\ref{sec:setting} and introduce the necessary concepts used for applications. In Section~\ref{sec:results}, we state our main results concerning the optimal transformation error and fluctuation-dissipation relation for incoherent and pure states, discuss their thermodynamic interpretation and apply them to three thermodynamic protocols of work extraction, information erasure and thermodynamically-free communication. The technical derivation of the main results can be found in Section~\ref{sec:math}. Finally, we conclude with an outlook in Section~\ref{sec:out}. This way we prove a general fluctuation-dissipation relation for thermodynamic distillation processes.

% ------------------------------------------------
% SECTION II - HIGH-LEVEL DESCRIPTION
% ------------------------------------------------

\section{High-level description}
\label{sec:high}

Before we give a formal statement of the setting studied in this paper, let us present a high-level description to give a flavour of our investigations. In our aims to identify fluctuation-dissipation phenomenon in the realm of resource-theoretic approach to thermodynamics, we shall extract the following main feature captured by the original works of Einstein and Smoluchowski. Namely, in order to obtain any ordered motion of a state of the system that is subjected to random forces, we necessarily need to dissipate energy that is proportional to fluctuations of energy induced by these random forces. On the other hand, the main point of the resource theory of thermodynamics is the question, whether one state can be thermodynamically transformed into another one. In our framework, we shall therefore examine what is the effect of the fluctuations present in the initial state of the system on the minimal amount of dissipation during a state transformation process. As we shall see at the end of this section, the proper fluctuating and dissipated quantity in the thermodynamic context will be given by the free energy of the system. 

As a warm up, we shall start now with a simple example, where the goal is to draw work from a given system (this is indeed an example of state transformation if one includes explicitly an ancillary weight system). More precisely, consider a model system with a continuous, non-degenerate energy spectrum with the ground state of energy $E_0=0$ that is prepared in a probabilistic mixture $\rho$ of different energy eigenstates $\ket{E}$ corresponding to energy $E$, i.e.,
\begin{equation}
    \rho=\int\limits_{0}^{\infty} p(E) \ketbra{E}{E}\, dE, 
\end{equation}
where $p(E)$ is a probability density function describing the system's distribution over energy levels. Our aim is now to use this model system with probabilistic (``fluctuating'') amounts of energy to make an almost deterministic (``ordered'') change of energy of another system. More formally, we are interested in performing $\epsilon$-deterministic work $W$, i.e., in changing the state of the ancillary weight system from one energy eigenstate $\ket{W_0}$ to another energy eigenstate $\ket{W_0+W}$ with probability $1-\epsilon$. 

How can we achieve this? If the distribution $p(E)$ is vanishing for $E\in[0,W]$, then we can couple the two systems and transfer the amount of energy $W$ between them by simply shifting the entire distribution $p(E)$ down by $W$, while moving the weight system up by $W$ (see the top panel of Fig.~\ref{fig:highleveldescript}). Similarly, if the bulk of $p(E)$ is localized far away from the ground energy 0, we can try to perform an analogous protocol, but this time we will fail with probability
\begin{equation}
    \epsilon=\int\limits_{0}^{W} p(E)\, dE,
\end{equation}
since the states with $E\in[0,W]$ cannot be lowered by $W$, as they would need to to be lowered below the ground state.

\begin{figure}[t]
    \centering
    \includegraphics{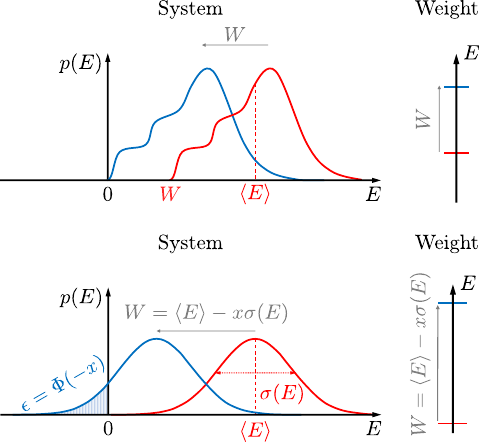}
    \caption{\label{fig:highleveldescript} \textbf{Transforming ``fluctuating'' to ``deterministic'' energy.} \textit{Top:} Despite fluctuations of energy, the lowest occupied state of the system is far away from the ground state, and so deterministic amount of work $W$ can drawn from it. However, $W$ is much smaller than the average energy content $\langle E\rangle$ of the system.\textit{ Bottom:} accepting probability of failure $\epsilon$, one can extract deterministic amount of work approaching the average energy $\langle E\rangle$, with the loss proportional to energy fluctuations $\sigma(E)$, where the proportionality constant is determined by $\epsilon$. Here, the initial distribution is assumed to be Gaussian with average $\langle E\rangle$ and standard deviation $\sigma(E)$.}
\end{figure}
To illustrate this more clearly, consider an important example of $p(E)$ given by a Gaussian distribution with a mean $\langle E\rangle$ and standard deviation $\sigma$:
\begin{equation}
    p(E) = \frac{1}{\sqrt{2\pi} \sigma} \exp\left(-\frac{1}{2}\frac{(E-\langle E\rangle)^2}{\sigma^2}\right).
\end{equation}
The importance of this example stems from the fact that in thermodynamics we are interested in total energy distributions of a large number $N$ of particles, and results like the central limit theorem tell us that the distributions of total quantities in such a case are approximated very well by Gaussian distributions. Of course, a Gaussian distribution is non-vanishing below the ground state energy 0, but as long the average energy $\langle E \rangle$ is far away from zero this can be neglected for the sake of our example. In the bottom panel of Fig.~\ref{fig:highleveldescript}, we present how shifting down a Gaussian distribution by its mean $\langle E\rangle$ decreased by a number of standard deviations $x\sigma$ results in an error $\epsilon=\Phi(-x)$, where $\Phi$ is the cumulative normal distribution function. Thus, for a fixed success probability of extracting work $W$ we can extract the average energy content of the system, $\langle E\rangle$, decreased by the quantity proportional to energy fluctuations $\sigma$. In other words, in order to transform the fluctuating type of energy into (almost) deterministic one, we need to lose (dissipate) some of it due to fluctuations. This simple scenario gives an intuition for why fluctuations may be related to dissipation.

Of course, the toy example we analysed above does not account for many features of realistic scenarios. First of all, it deals merely with mechanical work, whereas in thermodynamics one also has access to a thermal bath and can use it to draw even more thermodynamical work. Second, when considering systems of many particles we do not deal with non-degenerate spectrum, but rather at each energy we have a corresponding density of states. As a result, one may not be able to simply shift the distribution down, as there may be less low energy states then high energy states. Next, in the analysed example we only considered the protocol of work extraction, which is a very particular type of a general thermodynamic state transformation that physicists are interested in. Finally, since we deal with quantum mechanics, within each degenerate energy subspace we may deal with coherent superposition of states that can constructively or destructively interfere. Hence, the picture gets even more complex and requires a formalism that can account both for coherent and incoherent contributions to fluctuations. 

Despite these complications, in our work we prove that the original intuition from the simple toy example can be extended to general quantum thermodynamic scenarios with all the features described above. The crucial modification needed is the replacement of the concept of average energy and its fluctuations (relevant in the case of mechanical systems) by average free energy and its fluctuations (relevant for thermodynamic scenarios). To account for quantum systems prepared in arbitrary non-equilibrium states, one needs to use the non-equilibrium quantum generalisation of the classical expression for free energy, which also allows one to rigorously define the concept of free energy fluctuations. We prove that with these modifications in place, one can employ the intuition described above to investigate general thermodynamic distillation processes, which transform generic states with fluctuations of free energy into states with no free energy fluctuations (the equivalent of ``ordered energy'' states). More precisely, we show that for a fixed success probability of transformation $(1-\epsilon)$, during such a process the amount of free energy must be dissipated that is proportional to the initial free energy fluctuations (with the proportionality constant given in Eq.~\eqref{eq:a_intro}).

% ------------------------------------------------
% SECTION III - SETTING THE SCENE
% ------------------------------------------------

\section{Setting the scene}
\label{sec:setting}

% ------------------------------------------------
% SECTION III.A
% ------------------------------------------------

\subsection{Thermodynamic distillation processes}
\label{sec:distillation}

In order to formally define the thermodynamic distillation process, we first need to identify the set of thermodynamically-free states and transformations. By definition, a state of the system that is in equilibrium with a thermal environment $E$ at inverse temperature $\beta$ is a free state. Therefore, for a system described by a Hamiltonian $H$, the only free state is given by the thermal Gibbs state
\begin{equation}
\label{thermal_state}
    \gamma = \frac{e^{-\beta H}}{Z}, \qquad Z = \tr{e^{-\beta H}}.
\end{equation}
The set of free transformations that we consider is given by \emph{thermal operations}~\cite{Janzing2000, brandao2015second, horodecki2013fundamental}, which act on the system as
\begin{equation}
\label{eq:thermal_ops}
    \E(\rho)=\trr{E'}{U\left(\rho\otimes\gamma_E\right)U^{\dagger}},
\end{equation}	
where $U$ is a joint unitary acting on the system and the thermal environment $E$ that is described by a Hamiltonian $H_E$ and is prepared in a thermal Gibbs state $\gamma_E$ at inverse temperature $\beta$. Moreover, $U$ is commuting with the total Hamiltonian of the system and bath, \mbox{$[U, H\otimes \iden_E+ \iden\otimes H_E] = 0$}, and at the end we discard any subsystem $E'$ of the joint system of the considered system and environment.

A \emph{thermodynamic distillation} process is a thermodynamically free transformation from a general \emph{initial system} described by a Hamiltonian $H$ and prepared in a state $\rho$, to a \emph{target system} described by a Hamiltonian $\tilde{H}$ and in a state $\tilde{\rho}$ that is an eigenstate of $\tilde{H}$.\footnote{In fact, all of our results apply to a slightly more general setting with target states being proportional to the Gibbs state on their support, e.g. for \mbox{$\tilde{\rho}=\frac{\tilde{\gamma}_k}{\tilde{\gamma}_k+\tilde{\gamma}_l}\ketbra{\tilde{E}_k}{\tilde{E}_k}+\frac{\tilde{\gamma}_l}{\tilde{\gamma}_k+\tilde{\gamma}_l}\ketbra{\tilde{E}_l}{\tilde{E}_l}$}, where $\ket{\tilde{E}_i}$ denotes the eigenstate of $\tilde{H}$ and $\tilde{\gamma}_i$ is its thermal occupation.} An \emph{$\epsilon$-approximate thermodynamic distillation} process from $(\rho,H)$ to $(\tilde{\rho},\tilde{H})$ is a thermal operation that transforms the initial system $(\rho,H)$ to the \emph{final system} $(\rho_{\rm fin},\tilde{H})$ with $\rho_{\rm fin}$ being $\epsilon$ away from $\tilde{\rho}$ in the infidelity distance $\delta$,
\begin{equation}
        \delta(\rho_1,\rho_2):=1-\left(\mathrm{Tr}{\sqrt{\sqrt{\rho_1}{\rho_2}\sqrt{\rho_1}}}\right)^2.
\end{equation}
We say that $\rho$ is \emph{energy incoherent} if it is a convex combination of eigenstates of $H$. 

In this paper, we will study the distillation process from $N$ independent initial systems to arbitrary target systems, e.g., to $\tilde{N}$ independent target systems as illustrated in Fig.~\ref{fig:distillation}. In particular, we will be interested in the asymptotic behaviour for large $N$. Thus, our distillation setting is specified by a family of initial and target systems indexed by a natural number $N$. For each fixed $N$, the initial system $(\rho^N,H^N)$ consists of $N$ non-interacting subsystems with the total Hamiltonian $H^N$ and a state $\rho^N$ given by
\begin{equation}\label{eq:initial}
    H^N=\sum_{n=1}^N H^N_{n},\qquad \rho^N=\bigotimes_{n=1}^N \rho^N_{n}, 
\end{equation}
while the target system is described by an arbitrary Hamiltonian $\tilde{H}^N$ and a state $\tilde{\rho}^N=\ketbra{\tilde{E}_k^{N}}{\tilde{E}_k^{N}}$, with $\ket{\tilde{E}_k^{N}}$ being an eigenstate of $\tilde{H}^N$ corresponding to some energy $\tilde{E}_k^{N}$. Note that since $\tilde{H}^N$ is arbitrary, it does not need to describe $N$ particles; in fact, it can even be a Hamiltonian of a single qubit. 

A typical example of this setting is when initial and target systems are given by copies of independent and identical subsystems. More precisely, in this case, the family of initial systems is given by $H^N$ with $H^N_n=H$ and $\rho^N=\rho^{\otimes N}$, while the family of target systems is given by $\tilde{N}$ subsystems, each with a Hamiltonian $\tilde{H}$ and in a state $\ketbra{\tilde{E}_k}{\tilde{E}_k}$. One is then interested in the optimal distillation rate $\tilde{N}/N$ as $N$ tends to infinity. However, we will investigate a more general setting, allowing the subsystems to differ in both state and Hamiltonian, as long as the initial state is uncorrelated.

\begin{figure}[t]
    \centering
    \includegraphics{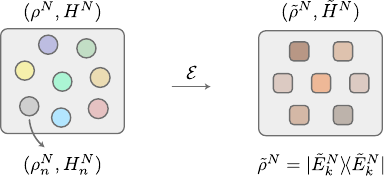}
    \caption{\label{fig:distillation} \textbf{Thermodynamic distillation process.} The arrow depicts the existence of a thermal operation transforming  $N$ independent initial systems to $\tilde{N}$ independent target systems. The circles and squares represent the initial and target systems with each subsystem described by a different Hamiltonian and prepared in a different state.} 
\end{figure}

% ------------------------------------------------
% SECTION III.B
% ------------------------------------------------

\subsection{Work extraction}
\label{sec:work-def}

One of the manifestations of the second law of thermodynamics is that for a system interacting with a bath in thermal equilibrium, the maximum amount of work that it can perform (that can be extracted from the system) is bounded by the difference $\Delta F$ between its initial and final free energy. Traditionally, the free energy $F = U-S/\beta$ has been defined only for states at thermal equilibrium, with $U$ denoting the internal energy and $S$ the entropy of the system. However, taking into account its operational meaning, one can extend its definition to investigate also the case of non-equilibrium states. More precisely, the relative entropy,
\begin{equation}
    D(\rho \|\gamma):=\tr{\rho (\log \rho-\log \gamma )},    
\end{equation}
can be interpreted as a non-equilibrium generalisation of the free energy difference between a state $\rho$ and a thermal state $\gamma$. It quantifies the maximum amount of work that can be extracted on average from the system in an out-of-equilibrium state~\cite{Esposito2009, Parrondo2015}.

Generally, work extraction protocols are based on controlling and changing the external parameters that define the Hamiltonian of the system~\cite{Alicki_1979, Aberg2013}. Within a resource-theoretic treatment~\cite{horodecki2013fundamental, Skrzypczyk2014}, however, we avoid using an external agent. Therefore, we explicitly model the ancillary \textit{battery system} $B$, intending to transform it from an initial pure energy state to another pure energy state with higher energy, see Fig.~\ref{fig:work_extraction}. A continuous Hamiltonian usually describes the battery, but we can as well choose a Hamiltonian with the discrete spectrum, as long as its energy differences coincide with the amount of work we want to extract. Without loss of generality, we focus on a two-level battery system described by a Hamiltonian $H^N_B$ with eigenstates $\ket{0}_B$ and $\ket{1}_B$ corresponding to energies $0$ and $W^N_{\rm ext}$, respectively. The possibility of extracting the amount of work equal to $W^N_{\rm ext}$ from $N$ subsystems described by a Hamiltonian $H^N$ and prepared in a state $\rho^N$ is then equivalent to the existence of a thermodynamic distillation process
\begin{equation}
    \E(\rho^N \otimes \ketbra{0}{0}_B) = \ketbra{1}{1}_B \,,    
\end{equation}
from $(N+1)$ initial subsystems described by a Hamiltonian $H^N+H^N_B$ to a target subsystems with a Hamiltonian~$H^N_B$. If only an $\epsilon$-approximate distillation with transformation error $\epsilon_N$ is possible, then $\epsilon_N$ directly measures the quality of extracted work, i.e., with probability $1-\epsilon_N$ we end up with a battery system in an excited state of energy $W^N_{\rm ext}$.

\begin{figure}[t]
    \centering
    \includegraphics{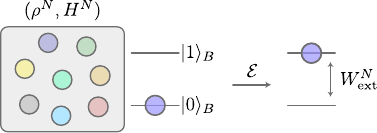}
    \caption{\label{fig:work_extraction} \textbf{Work extraction process.} Extraction of work $W^N_{\rm ext}$ from $N$ subsystems described by a Hamiltonian $H^N$ and prepared in a state $\rho^N$ can be seen as a particular case of thermodynamic distillation process $\E$ involving a battery system~$B$. The battery is modelled by a two-level system with energy levels $\ket{0}_B$ and $\ket{1}_B$ corresponding to energies $0$ and $W^N_{\rm ext}$, respectively. The initial system is given by the investigated $N$ subsystems with a battery in the ground state $\ket{0}_B$, while the target system is given just by the battery in the excited state~$\ket{1}_B$.
    }
\end{figure}

% ------------------------------------------------
% SECTION III.C
% ------------------------------------------------

\subsection{Information erasure}
\label{sec:erasure-def}

The connection between information and thermodynamics is as old as the thermodynamic theory itself, going back to the thought experiment known as Maxwell's demon~\cite{maxwell1872theory}. It suggests that if one has information about the particles' positions and momenta, one can reduce the entropy of a gas of particles without investing work, and thus violate the second law of thermodynamics. However, the recognition of the thermodynamic significance of information is perhaps best captured by the Szilard's engine~\cite{Szilard1929}, a simple setup that converts information into work. As in the Maxwell's demon example, the Szilard engine can overcome the second law of thermodynamics whenever some information about the state of the system is available. During the resolution of this puzzle, it became clear that thermodynamics imposes physical constraints on information processing. In particular, the second law can be reformulated as a statement that no thermodynamic process can result solely in the erasure of information. Every time information is erased, the erasure process is accompanied by a fundamental heat cost, i.e., an entropy increase in the environment~\cite{Bennett1982}.  Alternatively, the Landauer's Principle~\cite{Landauer1961} tells us that the erasure process has an unavoidable energetic cost, with the minimum possible amount of energy required to erase a completely unknown bit of information given by $\log 2 / \beta$ (see Ref.~\cite{AlickiH2019} in this context, where a more nuanced view on Szilard engine and Landauer erasure is presented).

Similarly to the case of work extraction, the erasure process can also be formulated as a particular type of thermodynamic distillation process. The $N$ bits of information that one wants to erase can be represented by $N$ two-level systems in a state $\rho^N$ with a trivial Hamiltonian. We also add the two-level battery system $B$ initially in an excited state $\ket{1}_B$ of energy $W^N_{\rm cost}$ to measure the energetic cost of erasure. Then, the erasure process resetting the state $\rho^N$ to a fixed state $\ket{0}^{\otimes N}$ is possible while investing $W^N_{\rm cost}$ work, if there exists the following distillation process:
\begin{equation}
    \E(\rho^N \otimes \ketbra{1}{1}_B) = \ketbra{0}{0}^{\otimes N}\otimes\ketbra{0}{0}_B \,,    
\end{equation}
with the initial and target Hamiltonians being identical. The transformation error quantifies the quality of erasure, and the process is illustrated in Fig.~\ref{fig:erasure}.
 
 \begin{figure}[t]
    \centering
    \includegraphics{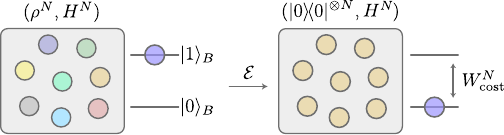}
    \caption{\label{fig:erasure} \textbf{Information erasure.} The $N$ bits of information to be erased are represented by $N$ subsystems in a state $\rho^N$ with a trivial Hamiltonian $H^N$. The process is performed by attaching a battery system $B$ in an excited state $\ket{1}_B$ with energy $W^N_{\rm cost}$, which measures the energetic cost of erasure. The erasure process resets the state $\rho^N$ to a fixed state $\ket{0}^{\otimes N}$, and de-excites the battery system.} 
\end{figure}

% ------------------------------------------------
% SECTION III.D
% ------------------------------------------------

\subsection{Thermodynamically-free communication}
\label{sec:encoding-def}

Since thermodynamics is closely linked with information processing, one can also study thermodynamic aspects of communication. A traditional communication scenario in which Alice wants to encode and transmit classical information to Bob over a quantum channel consists of the following three steps~\cite{wilde2013quantum}. First, she encodes a message \mbox{$m \in \{1, ..., M \}$} by preparing a quantum system in a state $\rho_m$. Then, she sends it to Bob via a noisy quantum channel~$\mathcal{N}$. Finally, Bob decodes the original message by performing an optimal measurement on $\mathcal{N}(\rho_m)$. Crucially, in this standard scenario, both Alice and Bob are completely unconstrained, meaning that they can employ all encodings and decodings for free, and the only thing beyond their control is the noisy channel~$\N$.

Recently, a modification of this scenario was introduced that allows one to quantify the thermodynamic cost of communication~\cite{narasimhachar2019quantifying,korzekwa2019encoding}. More precisely, it is assumed for simplicity that Alice and Bob are connected via a noiseless channel, and Bob's decoding is still unconstrained. However, Alice is constrained to \emph{thermodynamically-free encodings}, meaning that encoded states $\rho_m$ can only arise from thermal operations acting on a given initial state $\rho$, interpreted as an \emph{information carrier}. Physically, this means that Alice obeys the second law of thermodynamics, in the sense that the encoding channel is constrained to use no thermodynamic resources other than the ones initially present in the information carrier $\rho$. We illustrate this process in Fig.~\ref{fig:encoding}.

Now, the central question is: what is the optimal number of messages $M(\rho,\epsilon_{\mathrm{d}})$ that can be encoded into $\rho$ in a thermodynamically-free way, so that the average decoding error is smaller than $\epsilon_{\mathrm{d}}$? We will investigate the case when the information carrier is given by $N$ independent systems in a state $\rho^N$ and with a Hamiltonian $H^N$, as specified in Eq.~\eqref{eq:initial}. Then, instead of asking for $M(\rho^N,\epsilon_{\mathrm{d}})$, we can equivalently ask for the optimal encoding rate:
\begin{equation}
\label{eq:optimal-encoding-rate}
    R(\rho^N, \epsilon_{\mathrm{d}}) := \frac{\log [M(\rho^N, \epsilon_{\mathrm{d}})]}{N} \,.   
\end{equation}
As we will explain later in the paper, the optimal thermodynamically-free encodings (i.e., the ones that allow one to achieve the optimal rate $R$) can be chosen to be given by thermodynamic distillation processes. Through this connection and our results on optimal distillation processes, we will derive second-order asymptotic expansion of $R(\rho^N, \epsilon_{\mathrm{d}})$ working for large $N$.

\begin{figure}[t]
    \centering
    \includegraphics{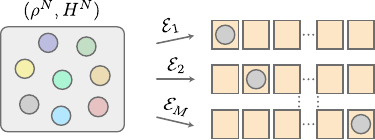}
    \caption{\label{fig:encoding}\textbf{Thermodynamically-free encoding.} The thermal encoding of information can be captured by a thermodynamic distillation process by considering $N$ independent subsystems in a state $\rho^N$ and with a Hamiltonian $H^N$ as an information carrier. The sender encodes a message \mbox{$m \in \{1, ..., M \}$} into it by applying a thermal operation $\E_{m}$ that transforms $\rho^N$ into mutually (almost) distinguishable states, and the receiver decodes the original message by performing a measurement on $\E_m(\rho^N)$.} 
\end{figure}

% ------------------------------------------------
% SECTION III.E
% ------------------------------------------------

\subsection{Information-theoretic notions and their thermodynamic interpretation}

Finally, before we proceed to present our results, let us introduce the necessary information-theoretic quantities together with their thermodynamic interpretation. For any $d$-dimensional quantum state~$\rho$, we define the relative entropy~$D$ between~$\rho$ and a thermal Gibbs state~$\gamma$, together with the corresponding relative entropy variance~$V$ and the function~$Y$ related to relative entropy skewness~\cite{Tomamichel2013,Chubb2018beyondthermodynamic,boes2020variance}:
\begin{subequations}
	\begin{align}
	D(\rho\|\gamma):=&\tr{\rho \left(\log \rho-\log\gamma\right)},\label{eq:D}\\
	V(\rho\|\gamma):=&\tr{\rho \left(\log \rho-\log\gamma-D(\rho\|\gamma)\right)^2},\label{eq:V}\\
	Y(\rho\|\gamma):=&\tr{\rho \left|\log \rho-\log\gamma - D(\rho\|\gamma)\right|^3}.\label{eq:W}
	\end{align}
\end{subequations}

It is clear from the above definitions that we are dealing with the average, variance and the absolute third moment of the ``random variable'' \mbox{$\log\rho-\log\gamma$}. As already mentioned, the average of this random variable, $D(\rho\|\gamma)$, can be interpreted as the non-equilibrium free energy of the system since
\begin{equation}
    \label{eq:rel_ent_gibbs}
    \frac{1}{\beta}D(\rho\|\gamma)= \tr{\rho H}-\frac{S(\rho)}{\beta}+\frac{\log Z}{\beta},
\end{equation} 
where
\begin{align}
    S(\rho):=&-\tr{\rho \log \rho}
\end{align}
is the von Neumann entropy. The higher moments can then be understood as fluctuations of the non-equilibrium free energy content of the system. This is most apparent for pure states $\rho=\ketbra{\psi}{\psi}$, as $V$ then simply describes energy fluctuations of the system:
\begin{align}
    \label{eq:energy_fluct}
    \frac{1}{\beta^2}V(\ketbra{\psi}{\psi}\|\gamma)&=\bra{\psi} H^2\ket{\psi}-\bra{\psi} H\ket{\psi}^2.
\end{align}
Moreover, as noted in Ref.~\cite{Chubb2018beyondthermodynamic}, when $\rho=\gamma'$ is a thermal distribution at some different temperature \mbox{$T'\neq T$}, the expression for $V$ becomes
\begin{equation}
\label{eq:var_capacity}
    V(\gamma'\|\gamma)=\left(1-\frac{T'}{T}\right)^2 \cdot\frac{c_{T'}}{k_B},
\end{equation}
where
\begin{equation}
    c_{T'}=\frac{\partial}{\partial T'} \tr{\gamma'H}
\end{equation}
is the specific heat capacity of the system in a thermal state at temperature $T'$, and $k_B$ is the Boltzmann constant. 

Now, for the initial system $(\rho^N,H^N)$, we introduce the following notation for free energy and free energy fluctuations:
    \begin{subequations}
    \begin{align}
        F^N := \frac{1}{\beta }\sum_{n=1}^N D(\rho^{N}_n\|\gamma^{N}_n),\label{eq:dn}\\
        \sigma^2(F^N) := \frac{1}{\beta^2 }\sum_{n=1}^N V(\rho^{N}_n\|\gamma^{N}_n),\label{eq:vn}\\ 
        \kappa^3(F^N) := \frac{1}{\beta^3 }\sum_{n=1}^N Y(\rho^{N}_n\|\gamma^{N}_n).\label{eq:yn}
    \end{align}
\end{subequations}
We also introduce
\begin{equation}
    \label{eq:deltaF}
    \Delta F^N := \frac{1}{\beta }\left(\sum_{n=1}^N D(\rho^N_n\|\gamma^N_n)-D(\tilde{\rho}^N\|\tilde{\gamma}^N)\right),
\end{equation}
which describes the free energy difference between the initial and target states, as well as
\begin{equation}
    \label{eq:F_diss}
    F^N_{\rm diss} := \frac{1}{\beta }\left(\sum_{n=1}^N D(\rho^N_n\|\gamma^N_n)-D(\rho^N_{\rm fin}\|\tilde{\gamma}^N)\right),
\end{equation}
which quantifies the amount of free energy that is dissipated in the distillation process, i.e., the free energy difference between the initial and final states. 

Let us also make two final technical comments. First, we only consider families of initial systems for which the limits of $\sigma^2(F^N)/N$ and $\kappa^3(F^N)/N$ as $N\rightarrow \infty$ are well-defined and non-zero. Second, in what follows, we will use a shorthand notation with $\simeq$, $\lesssim$ and $\gtrsim$ denoting equalities and inequalities up to terms of order $o(\sqrt{N})$.

% ------------------------------------------------
% SECTION IV - RESULTS
% ------------------------------------------------

\section{Results}
\label{sec:results}

% ------------------------------------------------
% SECTION IV.A
% ------------------------------------------------

\subsection{Optimal distillation error and fluctuation-dissipation relations} 
\label{sec:FDR-incoherent}

The first pair of our main results concerns thermodynamic distillation process from incoherent systems. The following two theorems connect the optimal distillation error to the free energy fluctuations of the initial state of the system, and the minimal amount of free energy dissipated in such a distillation process to these fluctuations.

\addtocounter{thm}{1}
\begin{subthm}[Optimal distillation error for incoherent states]
    \label{thm:incoherent}
    For a distillation setting with energy-incoherent initial states, the transformation error $\epsilon_N$ of the optimal $\epsilon$-approximate distillation process in the asymptotic limit is given by
    \begin{equation}
    \label{eq:error_incoherent}
        \lim_{N\to\infty}\epsilon_N = \lim_{N\to\infty}\Phi\left(-\frac{\Delta F^N}{\sigma(F^N)} \right),
    \end{equation}
    where $\Phi$ denotes the cumulative normal distribution function. Moreover, for any $N$ there exists an $\epsilon$-approximate distillation process with the transformation error $\epsilon_N$ bounded by
    \begin{align}
        \label{eq:error_bound}
        \epsilon_N&\leq \Phi\left(-\frac{\Delta F^N}{\sigma(F^N)} \right)+\frac{C\kappa^3(F^N)}{\sigma^3(F^N)},
    \end{align}
    where $C$ is a constant from the Berry-Esseen theorem that is bounded by
    \begin{equation}
        0.4097\leq C \leq 0.4748.
    \end{equation}
\end{subthm}

\addtocounter{thm}{1}
\addtocounter{subthm}{-1}
\begin{subthm}[Fluctuation-dissipation relation for incoherent states]
    \label{thm:incoherent2}
    The minimal amount of free energy dissipated in the optimal (minimising the transformation error $\epsilon$) distillation process from identical incoherent states asymptotically satisfies
    \begin{equation}
        \label{eq:diss}
        F^N_{\rm diss} \simeq a(\epsilon_N)\sigma(F^N),
    \end{equation}
    where 
    \begin{equation}
        \label{eq:a}
       a(\epsilon)=-\Phi^{-1}(\epsilon)(1-\epsilon)+\frac{\exp\Big(\frac{-(\Phi^{-1}(\epsilon))^2}{2}\Big)}{\sqrt{2\pi}}
    \end{equation}
    and $\Phi^{-1}$ is the inverse function of the cumulative normal distribution function $\Phi$.
\end{subthm}

We prove the above theorems in Secs.~\ref{sec:proof1a}~and~\ref{sec:proof2a}, and here we will briefly discuss their scope and consequences. We start by noting that combining Eqs.~\eqref{eq:error_incoherent}~and~\eqref{eq:diss} yields the optimal amount of dissipated free energy as a function of $\Delta F^N$ and $\sigma(F^N)$:
\begin{align}
    F^N_{\rm diss} \simeq &  \left(1-\Phi\left(-\frac{\Delta F^N}{\sigma(F^N)}\right)\right)\Delta F^N \nonumber\\
    &+\frac{\exp\left(-\frac{(\Delta F^N)^2}{2\sigma^2(F^N)}\right)}{\sqrt{2\pi}}\sigma(F^N).\label{eq:diss_delta_sigma}
\end{align}
Now, for the analysed case of independent initial subsystems, free energy fluctuations $\sigma(F^N)$ scale as $\sqrt{N}$. Thus, we can distinguish three regimes, depending on how the free energy difference between the initial and target states, $\Delta F^N$, behaves with growing $N$:
\begin{equation}
    \lim_{N\to\infty} \frac{\Delta F^N}{\sqrt{N}}=\left\{
    \begin{array}{l}
         \infty,  \\
         -\infty, \\
         \alpha\in \mathbb{R}.
    \end{array}
    \right.
\end{equation}

The first case corresponds to the target state having much smaller free energy than the initial state (as compared to the size of free energy fluctuations). According to Eq.~\eqref{eq:error_incoherent}, the transformation error then approaches zero in the asymptotic limit; while according to Eq.~\eqref{eq:diss_delta_sigma}, the amount of dissipated free energy $F^N_{\rm diss}\simeq \Delta F^N$, i.e., up to second order asymptotic terms the target and final states have the same free energy. This means that one can get arbitrarily close to the target state with much lower free energy than the initial state. The second case corresponds to the target state having much larger free energy than the initial state. The transformation error then approaches one in the asymptotic limit, while the amount of dissipated free energy $F^N_{\rm diss}$ approaches zero. This means that it is impossible to even get slightly closer to the target state with much higher free energy than the initial state, and so the optimal process corresponds to doing nothing (that is why there is no free energy dissipated). 

Finally, the third case that forms the essence of Theorems~\ref{thm:incoherent}~and~\ref{thm:incoherent2} corresponds to the target state having free energy very close to that of the initial state (again, the scale is set by the magnitude of free energy fluctuations). Our theorems then directly link the optimal transformation error and the minimal amount of dissipated free energy in the process to the free energy fluctuations of the initial state of the system. For two processes with the same free energy difference $\Delta F^N$, the process involving the initial state with smaller fluctuations will yield a smaller transformation error according to Eq.~\eqref{eq:error_incoherent}. Similarly, since the derivative of Eq.~\eqref{eq:diss_delta_sigma} over $\sigma(F^N)$ for a fixed $\Delta F^N$ is always positive, states with smaller free energy fluctuations will lead to smaller free energy dissipation. As a particular example consider a battery-assisted distillation process, i.e. a thermodynamic transformation from \mbox{$(\rho^N\otimes \ketbra{1}{1}_B,H^N+H_B)$} to \mbox{$(\tilde{\rho}^N\otimes \ketbra{0}{0}_B,H^N+H_B)$}, where the energy gap of the battery system $B$ is $W^N_{\mathrm{cost}}$. Now, the quality of transformation from $\rho^N$ to $\tilde{\rho}^N$ (measured by transformation error $\epsilon_N$) depends on the amount of work $W^N_{\mathrm{cost}}$ that we invest into the process. As expected, to achieve $\epsilon\leq 1/2$, we need to invest at least the difference of free energies $[D(\tilde{\rho}^N\|\tilde{\gamma}^N)-D(\rho^N\|\gamma^N)]/\beta$. However, Theorem~\ref{thm:incoherent} tells us how much more work is needed to decrease the transformation error to a desired level: the more free energy fluctuations there were in $\rho^N$, the more work we need to invest.

Let us also compare our two theorems to the results presented in Ref.~\cite{Chubb2018beyondthermodynamic}. There, the authors studied the incoherent thermodynamic interconversion problem between identical copies of the initial system, $\rho^{\otimes N}$, and identical copies of the target system, $\tilde{\rho}^{\otimes \tilde{N}}$. Here, for the price of the reduced generality of the target state (it has to be an eigenstate of the target Hamiltonian), we obtained a four-fold improvement. First, our result applies to general independent systems, not only to identical copies. Second, the Hamiltonians of the initial and target systems can vary, which is particularly important for applications like work extraction or thermodynamically-free communication. Third, we went beyond the second-order asymptotic result and found a single-shot upper bound on the optimal transformation error $\epsilon_N$, Eq.~\eqref{eq:error_bound}, that holds for any finite $N$. Thus, even in the finite $N$ regime, one can get a guarantee on the transformation error that is approaching the asymptotically optimal value as $N\to\infty$. Finally, we derived the expression for the actual amount of dissipated free energy in the optimal process and related it to the fluctuations of the free energy content of the initial state.

Our second pair of main results is analogous to the first pair, but concerns thermodynamic distillation process from $N$ identical copies of a pure quantum system. Thus, the following two theorems connect the optimal distillation error to the free energy fluctuations of the initial state of the system, and the minimal amount of free energy dissipated in such a distillation process to these fluctuations. To formally state these theorems, we need to introduce a technical notion of a Hamiltonian with incommensurable spectrum. Given any two energy levels, $E_i$ and $E_j$, of such a Hamiltonian, there does not exist natural numbers $m$ and $n$ such that $m E_i=n E_j$. We then have the following results.

\addtocounter{thm}{-1}
\begin{subthm}[Optimal distillation error for identical pure states]
    \label{thm:pure}
    For a distillation setting with $N$ identical initial systems, each in a pure state $\ketbra{\psi}{\psi}$ and described by the same Hamiltonian $H$ with incommensurable spectrum, the transformation error $\epsilon_N$ of the optimal $\epsilon$-approximate distillation process in the asymptotic limit is given by
    \begin{equation}
        \label{eq:error_pure}
            \lim_{N\to\infty}\epsilon_N = \lim_{N\to\infty}\Phi\left(-\frac{\Delta F^N}{\sigma(F^N)} \right),
    \end{equation}
    where $\Phi$ denotes the cumulative normal distribution function. Moreover, the result still holds if both the initial and target systems get extended by an ancillary system with an arbitrary Hamiltonian $H_A$, with the initial and target states being some eigenstates of $H_A$.
\end{subthm}

\addtocounter{thm}{1}
\addtocounter{subthm}{-1}
\begin{subthm}[Fluctuation-dissipation relation for identical pure states]
    \label{thm:pure2}
   The minimal amount of free energy dissipated in the optimal (minimising the transformation error $\epsilon$) distillation process from identical pure states asymptotically satisfies
    \begin{equation}
        \label{eq:diss_pure}
        F^N_{\rm diss} \gtrsim a(\epsilon_N)\sigma(F^N),
    \end{equation}
    where $a(\epsilon)$ is given by Eq.~\eqref{eq:a}.
\end{subthm}

We prove the above theorems in Secs.~\ref{sec:proof1b}~and~\ref{sec:proof2b}, while here we will only add one comment to the previous discussion. Namely, since for a pure state the free energy fluctuations are just the energy fluctuations (recall Eq.~\eqref{eq:energy_fluct}), and because in the considered scenario all pure states are identical, we have
\begin{equation}
    \label{eq:energy_fluc}
    \frac{1}{N}\sigma^2(F^N)=\langle H^2 \rangle_{\psi}-\langle{H}\rangle_{\psi}^2,
\end{equation}
where we use a shorthand notation \mbox{$\langle \cdot\rangle_\psi=\matrixel{\psi}{\cdot}{\psi}$}. Analogously to the incoherent case, the only non-trivial behaviour of the optimal transformation error happens when $\Delta F^N\simeq \alpha \sqrt{N}$, and its value is then specified by the ratio $\alpha/(\langle H^2 \rangle_{\psi}-\langle{H}\rangle_{\psi}^2)^{1/2}$.

% ------------------------------------------------
% SECTION IV.B
% ------------------------------------------------

\subsection{Optimal work extraction}
\label{sec:work-res}

As the first application of our results, we focus on work extraction process from a collection of $N$ non-interacting subsystems with Hamiltonians $H^N_n$ and in incoherent states $\rho^N_n$. As already described in Sec.~\ref{sec:work-def}, this is just a particular case of a thermodynamic distillation process. We only need to note that the pure battery state does not contribute to fluctuations $\sigma$ and $\kappa$, and that the difference between non-equilibrium free energies of the ground and excited battery states is just the energy difference $W^N_{\rm ext}$. Then, Theorem~\ref{thm:incoherent} tells us that, in the asymptotic limit, the optimal transformation error for extracting the amount of work $W^N_{\rm ext}$ is
\begin{equation}
\label{eq:workanderror}
    \lim_{N\to \infty} \epsilon_N = \lim_{N\to\infty}\Phi\left(\frac{W_{\rm ext}^N-F^N}{\sigma(F^N)}\right).
\end{equation}
We thus clearly see that again we have three cases dependent on the difference $(W_{\rm ext}^N-F^N)$. To get the asymptotic error different from zero and one, the extracted work $W^N_{\rm ext}$ has to be of the form
\begin{equation}
    W^N_{\rm ext}\simeq F^N - \alpha\sqrt{N},
\end{equation}
for some constant $\alpha$. Combining the above two equations yields the following second-order asymptotic expression for the extracted work:
\begin{equation}
    \label{eq:work}
    W_{\rm ext}^N\simeq F^N + \sigma(F^N)\Phi^{-1}(\epsilon).
\end{equation}
Thus, for a fixed quality of extracted work measured by~$\epsilon$, more work can be extracted from states with smaller free energy fluctuations (assuming that the initial free energy $F^N$ is fixed). This is a direct generalisation of the result obtained in Ref.~\cite{Chubb2018beyondthermodynamic} to a scenario with non-identical initial systems and with a cleaner interpretation of the error in the battery system. We present the comparison between our bounds and the numerically optimised work extraction processes in Fig.~\ref{fig:work_numerics}. 

\begin{figure}[t]
	\includegraphics[width=\columnwidth]{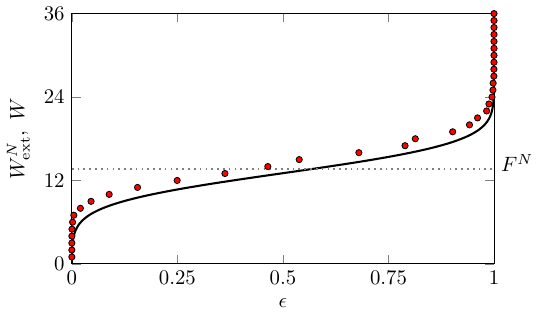}    
	\caption{\textbf{Optimal work extraction.} 
		Comparison between the asymptotic approximation, Eq.~\eqref{eq:work}, for the optimal amount of extracted work $W^N_{\mathrm{ext}}$ (solid black line) as a function of transformation error $\epsilon$, and the actual optimal value $W$ (red circles) obtained by explicitly solving the thermomajorisation conditions (see Sec.~\ref{sec:incoherent} for details). The inverse temperature of the thermal bath is chosen to be $\beta=1$, while the initial system is composed of 100 two-level subsystems. The first 59 subsystems are described by the Hamiltonian corresponding to a thermal state $0.6\ketbra{0}{0}+0.4\ketbra{1}{1}$, and the remaining 41 subsystems have the Hamiltonian leading to a thermal state $0.75\ketbra{0}{0}+0.25\ketbra{1}{1}$. The initial state of the system is given by 59 copies of a state $0.9\ketbra{0}{0}+0.1\ketbra{1}{1}$ and 41 copies of a state $0.7\ketbra{0}{0}+0.3\ketbra{1}{1}$. The non-equilibrium free energy of the total initial system, $F^N$, is indicated by a grey dotted line.
		\label{fig:work_numerics}}
\end{figure}

\begin{figure}[t]
    \includegraphics[width=\columnwidth]{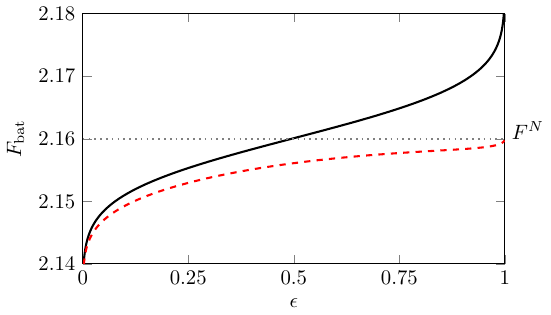}
	\caption{\label{fig:work_quality} \textbf{Work quality.} Non-equilibrium free energy $F_{\textrm{bat}}$ of the two-level battery system calculated for the final (dashed red line) and target (solid black line) state of the optimal work extraction process. The inverse temperature of the thermal bath is chosen to be $\beta = 1$, and the initial state that the work is extracted from is composed of 100 copies of a state 0.7\ketbra{0}{0}+0.3\ketbra{1}{1}. Each subsystem is described by the Hamiltonian corresponding to a thermal state 0.6\ketbra{0}{0}+0.4\ketbra{1}{1} and the non-equilibrium free energy of the total initial system, $F^N$, is indicated by a grey dotted line.} 
\end{figure}

Similarly, by employing Theorem~\ref{thm:pure}, we can investigate the optimal work extraction process from a collection of $N$ non-interacting subsystems with identical Hamiltonians $H$ and each in the same  pure state $\ketbra{\psi}{\psi}$. We simply need to choose the ancillary system $A$ to be the battery $B$ with energy splitting $W^N_{\mathrm{ext}}$ and the initial and target states to be given by $\ket{0}_B$ and $\ket{1}_B$. Also, since all systems are in identical pure states and have the same Hamiltonian, we have $\sigma(F^N)$ specified by Eq.~\eqref{eq:energy_fluc} and
\begin{align}
    \frac{1}{N}{F}^N&=\langle H\rangle_\psi+\frac{\log Z}{\beta}.
\end{align}
As a result, the optimal amount of work extracted from $N$ pure quantum systems up to second-order asymptotic expansion is given by:
\begin{equation}
    \label{eq:work_pure}
    \!\!\!W_{\rm ext}\!\simeq\! N\left(\!\langle H\rangle_\psi+\frac{\log Z}{\beta} + \frac{\langle H^2\rangle_\psi-\langle H\rangle_\psi^2}{\sqrt{N}}\Phi^{-1}(\epsilon)\!\right)\!.\!
\end{equation}

Finally, let us note that we can employ Theorem~\ref{thm:incoherent2} (and to some extent also Theorem~\ref{thm:pure2}) to investigate the meaning of work quality measured by the transformation error~$\epsilon$. So far we have measured the extracted work as the difference between the free energy of the initial battery's state and its target state that was obtained with success probability $1-\epsilon$. However, due to the aforementioned theorems, we know precisely the free energy of the actual final state of the battery, which can be used to quantify the actual amount of extracted work (with no error). In Fig.~\ref{fig:work_quality} we present the behaviour of both measures as a function of $\epsilon$, where it is clear that the two notions coincide for small error~$\epsilon$.

% ------------------------------------------------
% SECTION IV.C
% ------------------------------------------------

\subsection{Optimal cost of erasure}
\label{sec:erasure-res}

In order to obtain the optimal work cost of erasing $N$ two-level systems prepared in incoherent states $\rho^N_n$, we apply Theorem~\ref{thm:incoherent} analogously as in the previous section, but this time to the scenario described in Sec.~\ref{sec:erasure-def}. We then get the optimal transformation error in the erasure process given by
\begin{equation}
    \lim_{N\to \infty} \epsilon_N = \lim_{N\to\infty}\Phi\left(\frac{\frac{1}{\beta}S(\rho^N)-W^N_{\mathrm{cost}}}{\sigma(F^N)}\right),
\end{equation}
where $S(\rho^N)$ is the entropy of the initial state, and \mbox{$W^N_{\mathrm{cost}}$} is the invested work cost. Using analogous reasoning as in the case of work extraction, we can now obtain the second-order asymptotics for the cost of erasure:
\begin{equation}
    W^N_{\mathrm{cost}}\simeq \frac{S(\rho^N)}{\beta}  - \sigma(F^N)\Phi^{-1}(\epsilon).
\end{equation}

Let us make three brief comments on the above result. First, we only considered the application of the incoherent result, Theorem~\ref{thm:incoherent}, as in the case of trivial Hamiltonians, the erasure of a pure state $\ketbra{\psi}{\psi}^{\otimes N}$ is free (because all unitary transformations are then thermodynamically-free). Of course, our results straightforwardly extend to non-trivial Hamiltonians, but we believe that the simple case we described above is most illustrative and recovers the spirit of the original Landauer's erasure scenario. Second, since the maximally mixed initial state has vanishing free energy fluctuations, $\sigma(F^N)=0$, we cannot directly apply our result (that relates fluctuations of the initial state to dissipation) to get the erasure cost of $N$ completely unknown bits of information. However, using the tools described in Sec.~\ref{sec:math}, it is straightforward to show that in this case, the exact expression (working for all $N$) for the erasure cost is given by
\begin{equation}
    W^N_{\mathrm{cost}}= \frac{N}{\beta}\left( \log 2 - \frac{\log(1-\epsilon)}{N}\right).
\end{equation}
Thus, for the case of zero error one recovers the Landauer's cost of erasure~\cite{Alicki2004}. Third, analogously to the case of work extraction, here also Theorem~\ref{thm:incoherent2} can be employed to investigate the meaning of erasure quality quantified by $\epsilon$.

% ------------------------------------------------
% SECTION IV.D
% ------------------------------------------------

\subsection{Optimal thermodynamically-free communication rate}
\label{sec:encoding-res}

Finally, we now explain how Theorems~\ref{thm:incoherent}~and~\ref{thm:pure} allow one to obtain the optimal thermodynamically-free encoding rate into a collection of $N$ identical subsystems in either incoherent or pure states. We simply choose the target system to be a single $M$-dimensional quantum system with a trivial Hamiltonian $\tilde{H}=0$ that is prepared in any of the degenerate eigenstates of $\tilde{H}$. Note that the non-equilibrium free energy of such a target system is given by
\begin{equation}
    \frac{1}{\beta}D(\tilde{\rho}^N\|\tilde{\gamma}^N)=\frac{1}{\beta}\log M.
\end{equation}
Our theorems then tell us that in the asymptotic limit, the optimal transformation error $\epsilon$ in the considered distillation process is given by
\begin{equation}
    \lim_{N\to \infty} \epsilon_N = \lim_{N\to\infty}\Phi\left(\frac{\frac{1}{\beta}\log M-F^N}{\sigma(F^N)}\right).
\end{equation}
Rewriting the above, we get the following second-order asymptotic behaviour:
\begin{equation}
    \log M\simeq \beta F^N+\beta\sigma(F^N)\Phi^{-1}(\epsilon).
\end{equation}

Now, the distillation process above can be followed by unitaries that map between $M$ degenerate eigenstates of $\tilde{H}$ that we will simply denote $\ket{1},\dots,\ket{M}$. Crucially, note that such unitaries are thermodynamically-free because they act in a fixed energy subspace. Such a protocol then allows one to encode $M$ messages into $M$ states $\sigma_i$, each one being $\epsilon$-close in infidelity to $\ket{i}$ for $i\in\{1,\dots,M\}$. Decoding the message using a measurement in the eigenbasis of $\tilde{H}$ then leads to the average decoding error $\epsilon_{\mathrm{d}}$ satisfying:
\begin{equation}
    1-\epsilon_{\mathrm{d}}:=\frac{1}{M}\sum_{i=1}^M \matrixel{i}{\sigma_i}{i}=1-\epsilon,
\end{equation}
so that $\epsilon_{\mathrm{d}}=\epsilon$.

Using the communication protocol described above, we then get the following asymptotic lower bound on the optimal thermodynamically-free encoding rate into a state $\rho^N$ (recall Eq.~\eqref{eq:optimal-encoding-rate}):
\begin{equation}
    R(\rho^N,\epsilon_{\mathrm{d}}) \geq \frac{\beta}{N}(F^N+\sigma(F^N)\Phi^{-1}(\epsilon_{\mathrm{d}}))+o\left(\frac{1}{\sqrt{N}}\right).
\end{equation}
The above lower bound is exactly matching the upper bound for $R(\rho^N,\epsilon_{\mathrm{d}})$ recently derived in Ref.~\cite{korzekwa2019encoding} for a slightly different scenario with $\rho^N_n=\rho$ and $H^N_n=H$ for all $n$, with $\tilde{H}^N=H^N$, and with Gibbs-preserving operation instead of thermal operations. However, the proof presented there can be easily adapted to work in the current case if we keep the first restriction, i.e., when the initial state is $\rho^N=\rho^{\otimes N}$ and all initial subsystems have equal Hamiltonians. We explain in detail how to adapt that proof in Appendix~\ref{app:optimality}, where we also explain what technical result concerning hypothesis testing relative entropy needs to be proven in order to make the proof also work when subsystems are not identical. Here we conclude that
\begin{equation}
  \!\!\!R(\rho^{\otimes N},\epsilon_{\mathrm{d}}) \!=\! D(\rho\|\gamma)+\frac{\sqrt{V(\rho\|\gamma)}}{\sqrt{N}}\Phi^{-1}(\epsilon_{\mathrm{d}})+o\!\left(\!\frac{1}{\sqrt{N}}\!\right)\!,\!
\end{equation}
where $\rho$ is either a pure or incoherent state.

The above result can be thermodynamically interpreted as the inverse of the Szilard engine. While the Szilard engine converts bits of information into work, the protocol studied here employs the free energy of the system (i.e., the ability to perform work) to encode bits of information. While the asymptotic result was recently proven in Ref.~\cite{narasimhachar2019quantifying}, here we proved that this relation is deeper as it also connects fluctuations of free energy to the optimal average decoding error.

% ------------------------------------------------
% SECTION V - DERIVATION OF THE RESULTS
% ------------------------------------------------

\section{Derivation of the results}
\label{sec:math}

In what follows, we first introduce the mathematical formalism used to study the incoherent distillation process. We then use it to prove Theorems~\ref{thm:incoherent}~and~\ref{thm:incoherent2}. Finally, we also prove Theorems~\ref{thm:pure}~and~\ref{thm:pure2} by first mapping the problem of distillation from pure states to an equivalent incoherent problem, and then using the formalism of incoherent distillations.

% ------------------------------------------------
% SECTION V.A
% ------------------------------------------------

\subsection{Incoherent distillation process}
\label{sec:incoherent}

% ------------------------------------------------
% SECTION V.A.1
% ------------------------------------------------

\subsubsection{Distillation conditions via approximate majorisation}

A state of a $d$-dimensional quantum system $\rho$ will be called energy-incoherent if it commutes with the Hamiltonian of the system, i.e., when it is block-diagonal in the energy eigenbasis. Such a state can be equivalently represented by a $d$-dimensional probability vector $\v{p}$ given by the eigenvalues of $\rho$. Since the thermal Gibbs state $\gamma$ is energy-incoherent, it can be represented by a vector of thermal occupations $\v{\gamma}$. Moreover, an energy eigenstate $\ketbra{E_k}{E_k}$ can be represented by a \emph{sharp state} $\v{s}_k$, with $(\v{s}_k)_j=\delta_{jk}$.

In order to formulate the solution to the thermodynamic interconversion problem for incoherent states we will need two concepts: \emph{approximate majorisation} and \emph{embedding}. First, given two $d$-dimensional probability vectors $\v{p}$ and $\v{q}$, we say that $\v{p}$ majorises $\v{q}$, and write $\v{p} \succ \v{q}$, if and only if~\cite{bhatia1996matrix}
\begin{equation}
    \label{eq:majorisation}
    \forall  k : \:\: \sum_{j=1}^k p_j^\downarrow\geq \sum_{j=1}^k q_j^\downarrow,
\end{equation}
where $\v{p}^{\downarrow}$ denotes the vector $\v{p}$ rearranged in a decreasing order. Moreover, we say that $\v{p}$ $\epsilon$-post-majorises $\v{q}$~\cite{Chubb2018beyondthermodynamic}, and write $\v{p} \succ_\epsilon \v{q}$, if $\v{p}$ majorises $\v{r}$ which is $\epsilon$-close in the infidelity distance to $\v{q}$, i.e.,
\begin{equation}\label{eq:DefFidelity}
    1-F(\v{q},\v{r})\leq \epsilon, \qquad F(\v{q},\v{r}):=\left(\sum_{j=1}^d \sqrt{q_j r_j}\right)^2.   
\end{equation}
Second, we express the thermal distribution $\v{\gamma}$ as a probability vector with rational entries, 
\begin{align}
\label{eq:definitiongibbsvec}
	\v{\gamma}&=\left[\frac{D_1}{D},\dots,\frac{D_{d}}{D}\right],
\end{align}
with $D$ and $D_k$ being integers. Now, the embedding map is defined as a transformation that sends a $d$-dimensional probability distribution $\v{p}$ to a $D$-dimensional probability distribution $\hat{\v{p}}$ in the following way~\cite{brandao2015second}:
\begin{equation}
    \hat{\v{p}}=\left[\phantom{\frac{i}{i}}\!\!\!\!\right.\underbrace{\frac{p_1}{D_1},\dots,\frac{p_1}{D_1}}_{D_1\mathrm{~times}},
\,\dots\,,\underbrace{\frac{p_{d}}{D_{d}},\dots,\frac{p_{d}}{D_{d}}}_{D_{d}\mathrm{~times}}\left.\phantom{\frac{i}{i}}\!\!\!\!\right]. \label{eq:embedding}
\end{equation}
We will refer to the sets of repeated elements above as \emph{embedding boxes}. Observe that the embedded version of a thermal state $\v{\gamma}$ is a \emph{maximally mixed state} over $D$ states,
\begin{equation}
    \v{\eta}:=\frac{1}{D}[1,\dots,1],
\end{equation}
and the embedded version of a sharp state $\v{s}_k$ is a \emph{flat state} $\v{f}_k$ that is maximally mixed over a subset of $D_k$ entries, with zeros otherwise:
\begin{equation}
\!\!\hat{\v{s}}_{k}=\v{f}_{k}:=\left[\phantom{\frac{i}{i}}\!\!\!\!\right.\underbrace{0,\dots,0}_{\sum_{j=1}^{k-1}D_j},
\,\underbrace{1\dots 1}_{D_k}\,,\underbrace{0,\dots,0}_{\sum_{j=k+1}^d D_j}\left.\phantom{\frac{i}{i}}\!\!\!\!\right]\!. \label{eq:flat}
\end{equation}

We can now state the crucial theorem based on Ref.~\cite{brandao2015second} and concerning thermodynamic interconversion for incoherent states.
\begin{thm}[Corollary~7 of Ref.~\cite{Chubb2018beyondthermodynamic}]
	\label{thm:thermo_int}
	For the initial and target system with the same thermal distribution $\v{\gamma}$, there exists a thermal operation mapping between an energy-incoherent state $\v{p}$ and a state $\epsilon$-close to $\v{q}$ in infidelity distance, if and only if $\hat{\v{p}}\succ_\epsilon \hat{\v{q}}$.
\end{thm}

Despite the fact that in our case we want to study the general case of initial and final systems with different Hamiltonians, with a little bit of ingenuity we can still use the above theorem. Namely, we consider a family of total systems composed of the first $N$ subsystems with initial Hamiltonians $H^N_n$, and the remaining part described by the target Hamiltonian $\tilde{H}^N$. We choose initial states of the total system on the first $N$ subsystems to be a general product of incoherent states $\v{p}^N_n$, while the remaining part to be prepared in a thermal equilibrium state $\tilde{\v{\gamma}}^N$ corresponding to $\tilde{H}^N$. Since Gibbs states are free, this setting is thermodynamically equivalent to having just the first $N$ systems with Hamiltonians $H^N_n$ and in states $\v{p}_n^N$. Moreover, for target states of the total system, we choose thermal equilibrium states $\v{\gamma}^N_n$ for the first $N$ subsystems, and sharp states $\tilde{\v{s}}_{k}^N$ of the Hamiltonian $\tilde{H}^N$ for the remaining part. Again, this is thermodynamically equivalent to having just the system with Hamiltonian $\tilde{H}^N$ and in a state $\tilde{\v{s}}_{k}^N$. Thus, employing Theorem~\ref{thm:thermo_int}, an $\epsilon$-approximate distillation process for incoherent states exists if and only if:
\begin{equation}
    \left(\bigotimes_{n=1}^N \hat{\v{p}}^N_n \otimes \hat{\tilde{\v{\gamma}}}^N\right) \succ_\epsilon \left(\bigotimes_{n=1}^{N}\hat{\v{\gamma}}^N_n\otimes\hat{\tilde{\v{s}}}_{k}^N \right) .
\end{equation}
This way, using a single fixed Hamiltonian, we can encode transformations between different Hamiltonians.

Let us introduce the following shorthand notation:
\begin{equation}
    \hat{\v{P}}^N:=\bigotimes_{n=1}^N \hat{\v{p}}^{N}_n,\qquad \hat{\v{G}}^N:=\bigotimes_{n=1}^N {\hat{\v{\gamma}}}^N_n=\bigotimes_{n=1}^N {\v{\eta}}^N_n. \label{eq:PN_GN}
\end{equation}
Then, we can use the previous facts on the embedding map to conclude with the following statement: there exists an $\epsilon$-approximate thermodynamic distillation process from $N$ systems with Hamiltonians $H^N_n$ and in energy-incoherent states $\v{p}^N_n$ to a system with a Hamiltonian $\tilde{H}^N$ and in a sharp energy eigenstate $\tilde{\v{s}}_{k}^N$ if and only if
\begin{equation}
\label{eq:e-thermomajorisation}
    \hat{\v{P}}^N \otimes \tilde{\v{\eta}}^N \succ_{\epsilon} \hat{\v{G}}^N \otimes \tilde{\v{f}}^N_k.
\end{equation}

% ------------------------------------------------
% SECTION V.A.2
% ------------------------------------------------

\subsubsection{Information-theoretic intermission}
\label{sec:intermission}

Before we proceed, we need to make a short intermission for a few important comments concerning information-theoretic quantities introduced in Eqs.~\eqref{eq:D}-\eqref{eq:W}. For incoherent states $\rho$ and $\gamma$ represented by probability vectors $\v{p}$ and $\v{\gamma}$, these simplify and take the following classical form:
\begin{subequations}
	\begin{align}
	D(\v{p}\|\v{\gamma}):=&\sum_i p_i\left(\log \frac{p_i}{\gamma_i}\right),\\
	V(\v{p}\|\v{\gamma}):=&\sum_i p_i \left(\log \frac{p_i}{\gamma_i} - D(\v{p}\|\v{\gamma})\right)^2,\\
	Y(\v{p}\|\v{\gamma}):=&\sum_i p_i \left|\log \frac{p_i}{\gamma_i} - D(\v{p}\|\v{\gamma})\right|^3.
	\end{align}
\end{subequations}
Moreover, by direct calculation, one can easily show that the above quantities are invariant under embedding, i.e., \mbox{$D(\v{p}\|\v{\gamma})=D(\hat{\v{p}}\|\v{\eta})$}, and the same holds for $V$ and $Y$. Therefore
\begin{subequations}
	\begin{align}
	D(\v{p}\|\v{\gamma})=&D(\hat{\v{p}}\|\v{\eta})=\log D-H(\hat{\v{p}}),\label{eq:invariance1}\\
	V(\v{p}\|\v{\gamma})=&V(\hat{\v{p}}\|\v{\eta})=V(\hat{\v{p}}),\label{eq:invariance2}\\
	Y(\v{p}\|\v{\gamma})=&Y(\hat{\v{p}}\|\v{\eta})=Y(\hat{\v{p}}),
\end{align}
\end{subequations}
where
\begin{subequations}
	\begin{align}
	H(\v{p}):=&\sum_i p_i(-\log p_i),\\
	V(\v{p}):=&\sum_i p_i (\log p_i + H(\v{p}))^2,\\
	Y(\v{p}):=&\sum_i p_i \left|\log p_i + H(\v{p})\right|^3,\label{eq:invariance3}
\end{align}
\end{subequations}
and note that $V(\v{p})=0$ if and only if $\v{p}$ is a flat state.

% ------------------------------------------------
% SECTION V.A.3
% ------------------------------------------------

\subsubsection{Optimal error for a distillation process}

In order to transform the approximate majorisation condition from Eq.~\eqref{eq:e-thermomajorisation} into an explicit expression for the optimal transformation error, we start from the following result proven by the authors of Ref.~\cite{Chubb2018beyondthermodynamic}.
\begin{lem}[Lemma 21 of Ref.~\cite{Chubb2018beyondthermodynamic}]
	\label{lem:1shot-distill}
	Let $\bm p$ and $\bm q$ be distributions with $V(\bm q)=0$. Then
	\begin{align}
	\min\left\lbrace \epsilon \middle| \bm p\succ_\epsilon \bm q \right\rbrace = 1-\sum_{i=1}^{\exp H(\bm q)}p_i^{\downarrow}.
	\end{align}
\end{lem}
Applying the above lemma to Eq.~\eqref{eq:e-thermomajorisation} yields the following expression for the optimal error $\epsilon_N$:
\begin{equation}
    \label{eq:optimalesum}
    \epsilon_N = 1-\sum_{i=1}^{\exp [H(\hat{\v{G}}^{N})+H(\tilde{\v{f}}_k^N)]} \left(\hat{\v{P}}^N\otimes \tilde{\v{\eta}}^N\right)_i^\downarrow.
\end{equation}
Now, for an arbitrary distribution $\v{p}$ and any flat state~$\v{f}$, we make two observations: the size of the support of~$\v{f}$ is simply $\exp (H(\v{f}))$, and the entries of $\v{p}\otimes\v{f}$ are just the copied and scaled entries of $\v{p}$. As a result, the sum of the $l$ largest elements of $\v{p}$ can be expressed as
\begin{equation}
\label{eq:simpleobservation}
    \sum_{i=1}^l p_i^\downarrow = \sum_{i=1}^{l \exp (H(\v{f}))} (\v{p}\otimes \v{f})_i^\downarrow.
\end{equation}
Inverting the above expression we can write
\begin{equation}
\label{eq:simpleobservation2}
    \sum_{i=1}^{l} (\v{p}\otimes \v{f})_i^\downarrow = \sum_{i=1}^{l\exp (-H(\v{f}))} p_i^\downarrow,
\end{equation}
where the summation with non-integer upper limit $x$ should be interpreted as:
\begin{equation}
    \sum_{i=1}^x p_i :=\sum_{i=1}^{\lfloor x \rfloor} p_i +(x-\lfloor x \rfloor) p_{\lceil x \rceil}.
\end{equation}
Since $\tilde{\v{\eta}}$ is a flat state, we conclude that
\begin{equation}
\label{eq:optimal_e}
    \epsilon_N = 1- \sum_{i=1}^{\exp [H(\hat{\v{G}}^{N})+H(\tilde{\v{f}}_k^N)-H(\tilde{\v{\eta}}^N)]} (\hat{\v{P}}^N)_i^\downarrow.
\end{equation}

We see that the error depends crucially on partial ordered sums as above. To deal with these kind of sums, we introduce the function $\chi_{\v{p}}$ defined implicitly by the following equation
\begin{equation}
    \sum_{i=1}^{\chi_{\v{p}}(l)} p_i^\downarrow=\sum_i \{p_i | p_i\geq 1/l\}.
\end{equation}
In words: $\chi_{\v{p}}(l)$ counts the number of entries of $\v{p}$ that are larger than $1/l$. Now, we have the following lemma that will be crucial in proving our theorems.
\begin{lem}
    Every $d$-dimensional probability distribution $\v{p}$ satisfies the following for all $l\in\{1,\dots,d\}$ and for all $\alpha\geq 1$:
    \begin{subequations}
        \begin{align}
            \sum_{i=1}^{l} p_i^\downarrow&\geq \sum_{i=1}^{\chi_{\v{p}}(l)} p_i^\downarrow,\label{eq:lower}\\
            \sum_{i=1}^{l} p_i^\downarrow&\leq \sum_{i=1}^{\chi_{\v{p}}(\alpha l)/c} p_i^\downarrow,\label{eq:upper}
        \end{align}
    \end{subequations}
    where 
    \begin{equation}\label{eq:Expofc}
        c=\sqrt{\alpha}\sum_{i=\chi_{\v{p}}(\sqrt{\alpha}l)}^{\chi_{\v{p}}(\alpha l)} p_i^\downarrow.
    \end{equation}
    Moreover, as the probabilities are ordered in the sums given in Eq.~\eqref{eq:lower} and Eq.~\eqref{eq:upper}, it simply follows that
    \begin{equation}
        \label{chi_order}
        \chi_{\v{p}}(l) \leq l\leq \chi_{\v{p}}(\alpha l)/c.
    \end{equation}
\end{lem}
\begin{proof}
     The first inequality is very easily proven by observing that the number of entries larger than $1/l$, i.e., $\chi_{\v{p}}(l)$, is bounded from above by $l$ due to normalisation. Now, to prove the second inequality, we start from the following observation:
    \begin{equation}
        \sum_{i=1}^{\chi_{\v{p}}(\sqrt{\alpha}l)} \left(p_i^\downarrow-\frac{1}{\sqrt{\alpha}l}\right)\geq
        \sum_{i=1}^{\chi_{\v{p}}(\alpha l)}
        \left(p_i^\downarrow-\frac{1}{\sqrt{\alpha}l}\right),
    \end{equation}
    which comes from the fact that all the extra terms on the right hand side of the above are negative by definition. By rearranging terms we arrive at
    \begin{equation}
        \chi_{\v{p}}(\alpha l)-\chi_{\v{p}}(\sqrt{\alpha}l)\geq c l,
    \end{equation}
    which obviously implies
    \begin{equation}
        l \leq \frac{\chi_{\v{p}}({\alpha}l)}{c}.
    \end{equation}
\end{proof}

% ------------------------------------------------
% SECTION V.B
% ------------------------------------------------

\subsection{Proof of Theorem~\ref{thm:incoherent}}
\label{sec:proof1a}

The proof of Theorem~\ref{thm:incoherent} will be divided into two parts. First, we will derive the upper bound for the optimal transformation error $\epsilon_N$, Eq.~\eqref{eq:error_bound}. Then, we will provide a lower bound for $\epsilon_N$ and show that it is approaching the derived upper bound in the asymptotic limit, and so we will prove Eq.~\eqref{eq:error_incoherent}. 

% ------------------------------------------------
% SECTION V.B.1
% ------------------------------------------------

\subsubsection{Upper bound for the transformation error}

We start by introducing the following averaged entropic quantities for the total initial distribution $\hat{\v{P}}^N$:
\begin{subequations}\label{Moment_incoherent}
\begin{align}
    \!\! h_N:=&\frac{1}{N} H(\hat{\v{P}}^N)=\frac{1}{N}\sum_{n=1}^N H(\hat{\v{p}}^N_n)=:\frac{1}{N}\sum_{n=1}^N h^N_n,\label{eq:hn}\\
    \!\! v_N:=&\frac{1}{N} V(\hat{\v{P}}^N)=\frac{1}{N}\sum_{n=1}^N V(\hat{\v{p}}^N_n)=:\frac{1}{N}\sum_{n=1}^N v^N_n,\label{eq:vn2}\\
    \!\! y_N:=&\frac{1}{N} Y(\hat{\v{P}}^N)=\frac{1}{N}\sum_{n=1}^N Y(\hat{\v{p}}^N_n)=:\frac{1}{N}\sum_{n=1}^N y^N_n. \label{eq:yn2}
\end{align}
\end{subequations}
Note that the above $v_N$ and $y_N$ are, up to rescaling by $N/\beta$, incoherent versions of $\sigma^2(F^N)$ and $\kappa^3(F^N)$ defined in Eqs.~\eqref{eq:vn}-\eqref{eq:yn}. We also define the function~$l$:
\begin{align}\label{eq:Defn_f}
    l(z):=\exp\left(N h_N+z\sqrt{N v_N}\right).
\end{align}
We now rewrite the upper summation limit appearing in Eq.~\eqref{eq:optimal_e} employing the above function:
\begin{equation}
    \exp [H(\hat{\v{G}}^{N})+H(\tilde{\v{f}}^N_k)-H(\tilde{\v{\eta}}^N)]=l(x),
\end{equation}
so that
\begin{align}\label{eq:Def_x}
    x &= \frac{D(\hat{\v{P}}^N\|\hat{\v{G}}^{N})- D(\tilde{\v{f}}^N_k\|\tilde{\v{\eta}}^N)}{\sqrt{V(\hat{\v{P}}^N)}}.
\end{align}
This can be further transformed by employing the invariance of relative entropic quantities under embedding, Eqs.~\eqref{eq:invariance1}-\eqref{eq:invariance2}, leading to
\begin{align}
\label{eq:dist_cond}
    x&=\frac{\sum\limits_{n=1}^N D(\v{p}_{n}^N\|\v{\gamma}_n^N)-D(\tilde{\v{s}}_{k}^N\|\tilde{\v{\gamma}}^N)}{\left(\sum\limits_{n=1}^N V(\v{p}_{n}^N\|\v{\gamma}_{n}^N)\right)^{\frac{1}{2}}},
\end{align}
which is precisely the argument of $\Phi$ appearing in the statement of Theorem~\ref{thm:incoherent} in Eq.~\eqref{eq:error_incoherent}:
\begin{equation}\label{eq:x_continued}
    x = \frac{\Delta F^N}{\sigma(F^N)}.
\end{equation}
We conclude that with the above $x$ we can then rewrite the expression for the optimal transformation error, Eq.~\eqref{eq:optimal_e}, as
\begin{equation}
    \label{eq:error_exact}
    \epsilon_N = 1- \sum_{i=1}^{l(x)} (\hat{\v{P}}^N)_i^\downarrow.
\end{equation}

Next, we will find an upper bound for the error employing Eq.~\eqref{eq:lower}:
\begin{align}
\label{eq:first_bound}
    \epsilon_N&\leq 1-\sum_{i=1}^{\chi_{\hat{\v{P}}^N}(l(x))}(\hat{\v{P}}^N)_i^\downarrow\nonumber\\
    &=1- \sum_{i}\left\lbrace 
    \hat{P}^N_i \middle|
    \hat{P}^N_i \geq \frac{1}{l(x)}	\right\rbrace \!.
    \end{align}
In order to evaluate the above sum, consider $N$ discrete random variables $X_n$ taking values $-\log(\hat{\v{p}}^N_n)_i$ with probability $(\hat{\v{p}}^N_n)_i$, so that
\begin{subequations}
\begin{align}\label{eq:Moment_Inc}
    \langle X_n \rangle & = h^N_{n},\\
    \langle (X_n-\langle X_n \rangle)^2 \rangle & = v^N_{n},\\
    \langle\left| X_n-\langle X_n \rangle\right|^3\rangle & = y^N_{n},	\end{align}
\end{subequations}
where the average $\langle \cdot\rangle$ is taken with respect to the distribution $\hat{\v{p}}_{n}^N$. We then have the following
\begin{align}\label{eq: prob_on_set}
    &\sum_{i}\left\lbrace 
    \hat{P}^N_i \middle|
    \hat{P}^N_i \geq \frac{1}{l(x)}		
    \right\rbrace \nonumber\\
    &\quad=\sum_{i_1,\dots,i_N}\left\lbrace \prod_{n=1}^{N}
    (\hat{\v{p}}^N_{n})_{i_n} \middle|\prod_{n=1}^{N}(\hat{\v{p}}^N_{n})_{i_n}\geq \frac{1}{l(x)}	\right\rbrace \nonumber\\
    &\quad=\sum_{i_1,\dots,i_N}\left\lbrace \prod_{n=1}^{N}(\hat{\v{p}}^N_{n})_{i_n} \middle|-\sum_{n=1}^{N}\log (\hat{\v{p}}^N_{n})_{i_n}\leq \log l(x)	\right\rbrace \nonumber\\
    &\quad=\Pr\left[\sum_{n=1}^N X_n\leq Nh_N+x\sqrt{Nv_N}\right] \nonumber\\
    &\quad=\Pr\left[\frac{\sum_{n=1}^N (X_n-\langle X_n\rangle)}{\sqrt{\sum_{n=1}^N \langle (X_n-\langle X_n \rangle)^2 \rangle}}\leq x\right].
\end{align}
Now, the Berry-Esseen theorem~\cite{berry1941accuracy,esseen1942} tells us that
\begin{align}
\!\!\!\!   \left| \Pr\!\left[\!\frac{\sum_{n=1}^N (X_n-\langle X_n\rangle)}{\sqrt{\sum_{n=1}^N \langle (X_n\!-\!\langle X_n \rangle)^2 \rangle }}\leq x\!\right]\! -\! \Phi(x)\right|\!\leq \!\frac{C y_N}{\sqrt{Nv_N^3}},\!\!
\end{align}
where $C$ is a constant that was bounded in Refs.~\cite{Esseen1956, Shevtsova2011} by
\begin{equation}
\label{eq:c_bounds}
    0.4097\leq C \leq 0.4748.
\end{equation}
We thus have
\begin{equation}
    \label{eq:bound_Pn}
    \left|\sum_{i}\left\lbrace 
    \hat{P}^N_i \middle|
    \hat{P}^N_i \geq \frac{1}{l(x)}		
    \right\rbrace -  \Phi(x)\right| \leq \frac{Cy_N}{\sqrt{Nv_N^3}},
\end{equation}
and so we conclude that the error $\epsilon_N$ is bounded from above by
\begin{equation}
\label{eq:error_bound2}
    \epsilon_N\leq \Phi\left(-\frac{\Delta F^N}{\sigma(F^N)}\right)+\frac{C\kappa^3(F^N)}{\sigma^3(F^N)},
\end{equation}
which proves the single-shot upper bound on transformation error, Eq.~\eqref{eq:error_bound}, presented in Theorem~\ref{thm:incoherent}. Also, note that from Eq.~\eqref{eq:error_bound2}, it is clear that if $\lim_{N\to\infty} v_n$ and $\lim_{N\to\infty} y_n$ are well-defined and non-zero (as we assume), then
\begin{equation}
    \label{eq:asymptotic_upper}
    \lim_{N\rightarrow\infty} \epsilon_N \leq \lim_{N\rightarrow\infty} \Phi\left(-\frac{\Delta F^N}{\sigma(F^N)}\right).
\end{equation}

% ------------------------------------------------
% SECTION V.B.2
% ------------------------------------------------

\subsubsection{Lower bound for the transformation error}
\label{sec:incoherent_lower}

In order to lower bound the expression for the optimal error in the asymptotic limit, we choose \mbox{$\alpha=\exp(\delta\sqrt{N})$} with $\delta>0$ in Eq.~\eqref{eq:upper}. Thus, from Eq.~\eqref{chi_order} and Eq.~\eqref{eq:Defn_f} we have 
\begin{equation}\label{eq:Boundonlx}
    l(x)\leq \frac{\chi_{\hat{\v{P}}^N}(e^{\delta\sqrt{N}}l(x))}{c}=\frac{\chi_{\hat{\v{P}}^N}(l(x+\delta))}{c} ,
\end{equation}
where we can evaluate $c$ from Eq.~\eqref{eq:Expofc}:
\begin{align}\label{eq:c}
    \!\!\!\! c&=e^{\frac{\delta \sqrt{N}}{2}} \sum_{i=\chi_{\hat{\v{P}}^N}(l(x+\delta/2))}^{\chi_{\hat{\v{P}}^N}(l(x+\delta))} (\hat{\v{P}}^N)_i^\downarrow \nonumber\\
    \!\!\!\! &=e^{\frac{\delta \sqrt{N}}{2}} \left( \sum_{i}\left\lbrace 
    \hat{P}^N_i \middle|
    \hat{P}^N_i \geq \frac{1}{l(x+\delta)}	
    \right\rbrace  \right.\nonumber\\
    \!\!\!\! &\qquad\qquad\left. -\sum_{i}\left\lbrace 
    \hat{P}^N_i \middle|
    \hat{P}^N_i \geq \frac{1}{l(x+\delta/2)}
    \right\rbrace \right).
\end{align}
Using Eq.~\eqref{eq:bound_Pn} we can bound the above expression from below as
\begin{equation}
   c \geq e^{\frac{\delta \sqrt{N}}{2}}\left(\Phi(x+\delta)-\Phi(x+\delta/2)-\frac{2Cy_N}{\sqrt{Nv_N^3}}\right).
\end{equation}
Now, for any finite $\delta>0$ it is clear that there exists $N_0$ such that for all $N\geq N_0$ we have $c>1$. Combining with Eq.~\eqref{eq:Boundonlx}, for large enough $N$ we finally have 
\begin{equation}
    \label{eq:Boundonlx2}
    l(x)\leq \frac{\chi_{\hat{\v{P}}^N}(l(x+\delta))}{c} \leq \chi_{\hat{\v{P}}^N}(l(x+\delta)).
\end{equation}
Hence, using Eq.~\eqref{eq:Boundonlx2}, we have the following lower bound on transformation error
\begin{eqnarray}
     \label{eq:first_upper}
    \epsilon_N = 1-\sum_{i=1}^{l(x)}(\hat{\v{P}}^N)_i^\downarrow &\geq& 1- \!\!\sum_{i=1}^{\chi_{\hat{\v{P}}^N}(l(x+\delta))} \!\!(\hat{\v{P}}^N)_i^\downarrow\nonumber\\
    &=& 1- \sum_{i}\left\lbrace 
    \hat{P}^N_i \middle|
    \hat{P}^N_i \geq \frac{1}{l(x+\delta)}
    \right\rbrace \nonumber\\
    &\geq& 1-\Phi(x+\delta)-\frac{Cy_N}{\sqrt{Nv_N^3}},
\end{eqnarray}
where in the last line we used Eq.~\eqref{eq:bound_Pn} again. It is thus clear that
\begin{equation}
    \lim_{N\to\infty} \epsilon_N\geq 1-\lim_{N\to\infty} \Phi(x+\delta)=\lim_{N\to\infty} \Phi(-x-\delta)
\end{equation}
and, since it works for any $\delta>0$, we conclude that
\begin{equation}
    \lim_{N\to\infty} \epsilon_N\geq \lim_{N\to\infty} \Phi\left(-\frac{\Delta F^N}{\sigma(F^N)}\right).
\end{equation}
Combining the above with the bound obtained in Eq.~\eqref{eq:asymptotic_upper}, we arrive at
\begin{equation}\label{eq:err_opt}
    \lim_{N\to\infty}\epsilon_N=\lim_{N\to\infty} \Phi\left(-\frac{\Delta F^N}{\sigma(F^N)}\right),
\end{equation}
which proves the asymptotic expression for the transformation error, Eq.~\eqref{eq:error_incoherent}, presented in Theorem~\ref{thm:incoherent}.

% ------------------------------------------------
% SECTION V.C
% ------------------------------------------------

\subsection{Proof of Theorem~\ref{thm:incoherent2}}
\label{sec:proof2a}

The proof of Theorem~\ref{thm:incoherent2} will be divided into two parts. First, we will find the embedded version of the optimal final state minimising the dissipation of free energy $F^N_{\mathrm{diss}}$, and derive the expression for $F^N_{\mathrm{diss}}$ as a function of the initial state. Then, we will calculate $F^N_{\mathrm{diss}}$ up to second order asymptotic terms by upper and lower bounding the expression, and showing that the bounds coincide.

% ------------------------------------------------
% SECTION V.C.1
% ------------------------------------------------

\subsubsection{Deriving optimal dissipation}

We start by presenting an extension of Lemma~\ref{lem:1shot-distill} that not only specifies the optimal transformation error for approximate majorisation but also yields the optimal final state.

\begin{lem}\label{Opt_Diss_free}
    Let $\v{p}$ and $\v{q}$ be distributions of size $d$ with \mbox{$V(\v{q})=0$} and \mbox{$H(\v{q})=\log L$}. Then, states $\v{r}$ saturating $\v{p}\succ_\epsilon \v{q}$, i.e., such that \mbox{$\v{p}\succ\v{r}$} and $F(\v{q},\v{r})=1-\epsilon$ with $\epsilon$ being the minimal value specified by Lemma~\ref{lem:1shot-distill}, are given by $\Pi \v{q}^*$, where $\Pi$ is an arbitrary permutation,
    \begin{eqnarray}\label{eq:qstar}
     \v{q}^*&:=&
     \begin{cases}
       \frac{1-\epsilon}{L} &\quad \mathrm{for~} i\leq L,\\
       B \v{p}'_{i-L} &\quad\mathrm{for~}i>L,\\
     \end{cases}\, 
\end{eqnarray}
$B$ is an arbitrary $(d-L)\times(d-L)$ bistochastic matrix, and $\v{p}'$ is a vector of size $(d-L)$ with $p'_i=p^\downarrow_{i+L}$. Moreover, among all such distributions $\Pi\v{q}^*$, the entropy is minimised for the one with $B$ being the identity matrix.
\end{lem}

The proof of the above lemma can be found in Appendix~\ref{app:min_out}. Here, we apply it to the central Eq.~\eqref{eq:e-thermomajorisation} that specifies the conditions for the investigated $\epsilon$-approximate thermodynamic distillation process. As a result, the actual total final state in the embedded picture, $\hat{\v{F}}^N$, which is $\epsilon$ away from the target state $\hat{\v{G}}^N \otimes \tilde{\v{f}}^N_k$, is given, up to permutations, by
\begin{eqnarray}
\label{eq:distribution-total}
     \hat{\v{F}}^N&=&
     \begin{cases}
       \frac{1-\epsilon}{K} &\quad \mathrm{for~} i\leq K,\\
       (\hat{\v{P}}^N\otimes\tilde{\v{\eta}}^N)_{i}^\downarrow &\quad\mathrm{for~}i>K.\\
     \end{cases}\, 
\end{eqnarray}
where
\begin{equation}
    K=\exp(H(\hat{\v{G}}^N)+H(\tilde{\v{f}}^N_k)).% = D \tilde{D}.
\end{equation}
Note that we have chosen $\hat{\v{F}}^N$ to minimise entropy since, due to Eq.~\eqref{eq:invariance1}, it translates into the real (unembedded) final state $\v{F}^N$ with maximal free energy, i.e., it leads to minimal free energy dissipation.

The next step consists of going from the embedded to the unembedded picture. Importantly, if a state in the embedded picture is not uniform within each embedding box, the unembedding will effectively lead to the loss of free energy. However, we note that the final state is given by Eq.~\eqref{eq:distribution-total} only up to permutations. Thus, one may freely rearrange its elements to minimise such a loss. In particular, in Appendix~\ref{app:DDF} we show that for the case of identical initial states (i.e., $\v{P}^n=\v{p}^{\otimes N}$), there exists a permutation that transforms $\hat{\v{F}}^N$ so that it is uniform in almost all embedding boxes, which leads to exponentially small dissipation of free energy (i.e., no dissipation up to second-order asymptotics).  Therefore, employing the definition of dissipated free energy from Eq.~\eqref{eq:F_diss} and the above discussion, we have
\begin{align}
    F^N_{\mathrm{diss}}=&\frac{1}{\beta}\left( D(\v{P}^N\|\v{G}^N)-D(\v{F}^N\|\v{G}^N\otimes \tilde{\v{G}}^N)\right)\nonumber\\
    \simeq&\frac{1}{\beta}\left( H(\hat{\v{F}}^N)-H(\hat{\v{P}}^N)-\log \tilde{D}\right).
\end{align}
Here, $\tilde{D}$ is the embedding constant defined by $\tilde{\v{G}}^N$ according to Eq.~\eqref{eq:definitiongibbsvec}, and similarly $D$ will denote this embedding constant for $\v{G}^N$.
Note, however, that these can always be chosen to be equal, since the only thing that matters is that $\tilde{G}^N_k=\tilde{D}_k/\tilde{D}$ and ${{G}}^N_k={D}_k/D$, i.e., the change in $\tilde{D}$ or $D$ can be compensated by the appropriate change in $\tilde{D}_k$ or $D_k$.

Next, noting that $K=D\tilde{D}_k=\tilde{D}\tilde{D}_k$, we calculate the entropy of $\hat{\v{F}}^N$:
\begin{align}
    H(\hat{\v{F}}^N) =& -\sum_{i=1}^K \frac{1-\epsilon}{K}\log\left(\frac{1-\epsilon}{K}\right) \nonumber\\
    &- \sum_{i > \tilde{D}_k} (\hat{\v{P}}^N)_i{^{\downarrow}} \log \frac{(\hat{\v{P}}^N)_i{^{\downarrow}}}{\tilde{D}}\nonumber\\
    =&-(1-\epsilon)\log(1-\epsilon)+(1-\epsilon)\log K\nonumber\\
    & - \sum_{i > \tilde{D}_k} (\hat{\v{P}}^N)_i{^{\downarrow}} \log (\hat{\v{P}}^N)_i{^{\downarrow}} +\epsilon \log \tilde{D}\nonumber\\
    \simeq & \log \tilde{D} +(1-\epsilon)\log\tilde{D}_k\nonumber\\
    &- \sum_{i > \tilde{D}_k} (\hat{\v{P}}^N)_i{^{\downarrow}} \log (\hat{\v{P}}^N)_i{^{\downarrow}},
\end{align}
where we have dropped the term $(1-\epsilon) \log (1-\epsilon)$ as it is constant (i.e. it does not scale with $N$). 
Therefore, the dissipated free energy in the optimal distillation process is simply given by
\begin{equation}
    \label{eq:dissipatedfreeenergyapprox}
   \!\! F^N_{\rm{diss}} \simeq\frac{1}{\beta}\!\left(\! (1-\epsilon)\log \tilde{D}_k +\sum_{i = 1}^{\tilde{D}_k} (\hat{\v{P}}^N)_i{^{\downarrow}} \log (\hat{\v{P}}^N)_i{^{\downarrow}} \!\right)\!.\!
\end{equation}

% ------------------------------------------------
% SECTION V.C.2
% ------------------------------------------------

\subsubsection{Calculating optimal dissipation}

We now proceed to provide bounds for $F^N_{\rm{diss}}$. To do so, we first make use of Eq.\eqref{chi_order} and the fact that $\log x$ is negative for all $x\in(0,1)$ to write 
\begin{align}
\label{eq:ineq}
    \!\!\sum_{i = 1}^{\tilde{D}_k} (\hat{\v{P}}^N)_i{^{\downarrow}} \log (\hat{\v{P}}^N)_i{^{\downarrow}} \leq \sum_{i = 1}^{\chi_{\hat{\v{P}}^N}(\tilde{D}_k)} (\hat{\v{P}}^N)_i{^{\downarrow}} \log (\hat{\v{P}}^N)_i{^{\downarrow}} \!.\!
\end{align}
The right hand side of the above can then be recast as follows,
\begin{align}
&\!\!\sum_{i = 1}^{\chi_{\hat{\v{P}}^N}(\tilde{D}_k)} (\hat{\v{P}}^N)_i{^{\downarrow}} \log (\hat{\v{P}}^N)_i{^{\downarrow}}\nonumber \\
&~=\sum_{i}\left\lbrace 
    \hat{P}^N_i\log (\hat{P}^N_i) \middle|
    \hat{P}^N_i \geq \frac{1}{\tilde{D}_k}		
    \right\rbrace \nonumber \\
    &~=\!\!\sum_{i_1,\dots,i_N}\!\left\lbrace \prod_{n=1}^{N}
    (\hat{\v{p}}^N_{n})_{i_n}\sum_{m=1}^N\log(\hat{\v{p}}^N_{m})_{i_m} \middle|\prod_{n=1}^{N}(\hat{\v{p}}^N_{n})_{i_n}\geq \frac{1}{\tilde{D}_k}\!\right\rbrace  \nonumber\\
    &~=\!\!\sum_{i_1,\dots,i_N}\!\left\lbrace \prod_{n=1}^{N}
    (\hat{\v{p}}^N_{n})_{i_n}\sum_{m=1}^N\log(\hat{\v{p}}^N_{m})_{i_m}\right. \nonumber\\
    &\qquad\qquad\qquad\qquad \left| -\sum_{m=1}^N\log(\hat{\v{p}}^N_{m})_{i_m}\leq \log \tilde{D}_k \right\rbrace.
    \label{eq:PNlogPN}
\end{align}
Let us now note that from the definition in Eq.~\eqref{eq:deltaF} we have (here we employ the notation introduced in Eqs.~\eqref{eq:hn}-\eqref{eq:yn2}):
\begin{align}
    \beta\Delta F^N &= \sum\limits_{n=1}^N D(\v{p}_{n}^N\|\v{\gamma}_n^N)-D(\tilde{\v{s}}_{k}^N\|\tilde{\v{\gamma}}^N)\nonumber\\
    &= H(\tilde{\v{f}}_{k}^N) - \sum\limits_{n=1}^N H(\hat{\v{p}}_{n}^N)=\log \tilde{D}_k - N h_N .
\end{align}
Next, since for non-trivial error we need
\begin{equation}
\label{eq:xsigmF}
    \beta\Delta F^N = x\beta\sigma(F^N) = x \sqrt{N v_N},
\end{equation}
with some constant $x$, we can rewrite $\log \tilde{D}_k$ as
\begin{equation}
    \log \tilde{D}_k = N h_N +x \sqrt{N v_N}. 
    \label{eq:Dktilde}
\end{equation}

Coming back to Eq.~\eqref{eq:PNlogPN}, we can rewrite it by introducing $N$ discrete random variables $\{X_n\}_{n=1}^N$ assuming values $-\log (\hat{\v{p}}_n^N)_{i_n}$ with probability $(\hat{\v{p}}^N_n)_{i_n}$. Crucially, since the sum of their averages is $Nh_N$ and their total variance is $Nv_N$, the condition in the summation in Eq.~\eqref{eq:PNlogPN} simply becomes:
\begin{equation}
    \label{eq:constraint}
    Y_N:=\frac{\sum_{n=1}^N(X_n-\langle X_n\rangle)}{\sqrt{\sum_{n=1}^N\langle(X_n-\langle X_n\rangle)^2\rangle}} \leq x.
\end{equation}
Thus, we can write Eq.~\eqref{eq:PNlogPN} in a compact way as follows
\begin{align}
    &\sum_{i = 1}^{\chi_{\hat{\v{P}}^N}(\tilde{D}_k)} (\hat{\v{P}}^N)_i{^{\downarrow}} \log (\hat{\v{P}}^N)_i{^{\downarrow}}=-\int_T \sum_{n=1}^N
    X_N  d P, 
\end{align}
where $P$ is discrete probability measure given by $P(i_1,\ldots,i_N)=\prod_{n=1}^{N} (\hat{\v{p}}^N_{n})_{i_n}$ and $T$ is a region satisfying the constraint $Y_N<x$. Noting that 
\begin{align}
    \sum_{n=1}^N X_n= \sqrt{N v_n} Y_N + N h_N
\end{align}
we can further rewrite it as 
\begin{align}
    \!\!\!\int\limits_T \sum_{n=1}^N
    X_N  d P=
    \sqrt{Nv_N}\!\!\int\limits_{Y_N\leq x}\!\!  
    Y_N dP   + N h_N \!\!\int\limits_{Y_N\leq x}\!\! dP.\!
\end{align}
The second integral on the right hand side of the above was already calculated and is equal to $1-\epsilon$, where $\epsilon$ is the optimal transformation error from Theorem~\ref{thm:incoherent}. Further, since $Y_N$ is a standarized sum of independent random variables, its distribution tends to a normal Gaussian distribution with density denoted by  $\phi(x)$, i.e.,
\begin{equation}
    \int\limits_{Y_N\leq x}\!\!  
    Y_N dP\simeq\int_{-\infty}^x  y \phi(y) d y = -\frac{e^{-x^2/2}}{\sqrt{2\pi}}.
\end{equation}
We thus obtain: 
\begin{align}
    \!\!\!&\!\!\!\sum_{i = 1}^{\chi_{\hat{\v{P}}^N}(\tilde{D}_k)} (\hat{\v{P}}^N)_i{^{\downarrow}} \log (\hat{\v{P}}^N)_i{^{\downarrow}}\nonumber\\
    \!\!\!&\!\!\!~  \simeq \sqrt{N v_N}\frac{e^{-x^2/2}}{\sqrt{2\pi}} - N h_N(1-\epsilon)\nonumber\\
    \!\!\!&\!\!\!~  = \sqrt{Nv_N}\frac{e^{-x^2/2}}{\sqrt{2\pi}} -(1-\epsilon) (\log \tilde{D}_k-x\sqrt{Nv_N})\nonumber\\
    \!\!\!&\!\!\!~  = \beta\sigma(F^N)\!\left(\!\frac{e^{-x^2/2}}{\sqrt{2\pi}}+x(1-\epsilon)\!\right)\! -\!(1-\epsilon) \log \tilde{D}_k.
\end{align}

We can now use the inequality from Eq.~\eqref{eq:ineq} and substitute the above to Eq.~\eqref{eq:dissipatedfreeenergyapprox} to arrive at:
\begin{equation}
   F^N_{\rm{diss}} \lesssim \sigma(F^N)\left(\frac{e^{-x^2/2}}{\sqrt{2\pi}}+x(1-\epsilon)\right).
\end{equation}
Next, employing Eq.~\eqref{eq:xsigmF} and the expression for optimal transformation error from Theorem~\ref{thm:incoherent}, we can re-express $x$ as
\begin{equation}
    x=-\Phi^{-1}(\epsilon)
\end{equation}
to finally obtain
\begin{equation}    
\label{eq:diss_inc_upper}
   F^N_{\rm{diss}} \lesssim \sigma(F^N)\left(\frac{e^{\frac{-(\Phi^{-1}(\epsilon))^2}{2}}}{\sqrt{2\pi}}-\Phi^{-1}(\epsilon)(1-\epsilon)\right).
\end{equation}

To provide a lower bound of $F^N_{\rm diss}$ given in Eq.~\eqref{eq:dissipatedfreeenergyapprox}, we simply follow the argument we have given in Sec.~\ref{sec:incoherent_lower}. From Eq.~\eqref{eq:upper} it straightforwardly follows that 
\begin{align}
    \!\sum_{i=1}^{\tilde{D}_k} (\hat{\v{P}}^N)_i^\downarrow\log(\hat{\v{P}}^N)_i^\downarrow &\geq\!\!\! \sum_{i=1}^{\frac{\chi_{\hat{\v{P}}^N}(\tilde{D}_k e^{\delta\sqrt{N}})}{c}} \!\!\!\!\!(\hat{\v{P}}^N)_i^\downarrow\log(\hat{\v{P}}^N)_i^\downarrow,\!
\end{align}
where $c$, as before, can be lower bounded by 1 for large enough $N$. We thus get
\begin{equation}
\!\!\! F^N_{\mathrm{diss}} \gtrsim\frac{1}{\beta}\left(\! (1-\epsilon)\log \tilde{D}_k +\!\!\!\!\!\!\!\!\!\!\!\sum_{i=1}^{\chi_{\hat{\v{P}}^N}(\tilde{D}_k e^{\delta\sqrt{N}})} \!\!\!\!\!\!\!\!\!\!\!(\hat{\v{P}}^N)_i^\downarrow\log(\hat{\v{P}}^N)_i^\downarrow\!\right)\!.\! \!
\end{equation}
Since the above inequality holds for any $\delta>0$, therefore the limit $\delta\rightarrow 0$ we have
\begin{equation}
   F^N_{\mathrm{diss}} \gtrsim \sigma(F^N)\left(\frac{e^{\frac{-(\Phi^{-1}(\epsilon))^2}{2}}}{\sqrt{2\pi}}-\Phi^{-1}(\epsilon)(1-\epsilon)\right).
\end{equation}
Combining the above with the upper bound from Eq.~\eqref{eq:diss_inc_upper} we finally arrive at
\begin{equation}
      F^N_{\rm{diss}}\simeq a(\epsilon) \sigma(F^N), 
\end{equation}
where $a(\epsilon)$ is given by Eq.~\eqref{eq:a}.

% ------------------------------------------------
% SECTION V.D
% ------------------------------------------------

\subsection{Proof of Theorem~\ref{thm:pure}}
\label{sec:proof1b}

The proof of Theorem~\ref{thm:pure} will be divided into three parts. First, we will show that a thermodynamic distillation process from a general state $\rho$ can be reduced to a distillation process from an incoherent state that is a dephased version of $\rho$. Employing this observation, we will recast the problem under consideration in terms of approximate majorisation and thermomajorisation as described in Sec.~\ref{sec:incoherent}. Then, in the second part of the proof, we will derive the upper bound for the optimal transformation error $\epsilon_N$. Finally, in the third part, we will provide a lower bound for $\epsilon_N$ and show that it is approaching the derived upper bound in the asymptotic limit, and so we will prove Eq.~\eqref{eq:error_pure}. 

% ------------------------------------------------
% SECTION V.D.1
% ------------------------------------------------

\subsubsection{Reducing the problem to the incoherent case}

The thermodynamic distillation problem under investigation is specified as follows. The family of initial systems consists of a collection of $N$ identical subsystems, each with the same Hamiltonian
\begin{equation}
    H=\sum_{i=1}^d E_i \ketbra{E_i}{E_i},
\end{equation}
and an ancillary system with an arbitrary Hamiltonian~$H_A$ (note that the ancillary system can always be ignored by simply choosing its dimension to be 1). The family of initial states is given by
\begin{equation}\label{eq:Initial_pure_state_tensor}
    \rho^N= \psi^{\otimes N} \otimes \ketbra{E^A_0}{E^A_0},
\end{equation}
where
\begin{equation}
    \psi=\ketbra{\psi}{\psi},\quad \ket{\psi}=\sum_{i=1}^d \sqrt{p_i} e^{i\phi_i} \ket{E_i},
\end{equation}
is an arbitrary pure state and $\ket{E^A_0}$ is an eigenstate of $H_A$ with energy $E^A_0$. The family of target systems is composed of subsystems described by arbitrary Hamiltonians~$\tilde{H}^N$ and a subsystem described by the Hamiltonian~$H_A$. The family of target states is given by 
\begin{equation}
    \tilde{\rho}^N= \ketbra{\tilde{E}^{N}_k}{\tilde{E}^{N}_k} \otimes \ketbra{E^A_1}{E^A_1},
\end{equation}
where $\ket{\tilde{E}^{N}_k}$ is some eigenstate of $\tilde{H}^N$ and $\ket{E^A_1}$ is an eigenstate of $H_A$ with energy $E^A_1$. We are thus interested in the existence of a thermal operation $\E$ approximately performing the following transformation:
\begin{equation}
    \psi^{\otimes N} \otimes \ketbra{E^A_0}{E^A_0} ~\xrightarrow{~\E~}~ \ketbra{\tilde{E}^{N}_k}{\tilde{E}^{N}_k} \otimes \ketbra{E^A_1}{E^A_1}.
\end{equation}

We now have the following simple, but very useful, lemma.
\begin{lem}
	\label{lem:pure-to-incoherent}
   Every incoherent state $\sigma$ achievable from a state $\rho$ through a thermal operation is also achievable from $\D(\rho)$, where $\D$ is the dephasing operation destroying coherence between different energy subspaces:
	\begin{align}
	\exists \E: \E(\rho) = \sigma \quad\Leftrightarrow \quad \E(\D(\rho)) = \sigma.
	\end{align}
\end{lem}
\begin{proof}
First, for a given $\rho$ and incoherent $\sigma$, assume that there exists a thermal operation $\E$ such that $\E(\rho) = \sigma$. Now, employing the fact that every thermal operation is covariant with respect to time-translations~\cite{lostaglio2015description}, and using the fact that incoherent $\sigma$ by definition satisfies $\mathcal{D}(\sigma) = \sigma$, we get 
\begin{equation}
\label{dephasingrho}
    \E(\D(\rho))=\D(\E(\rho))=\D(\sigma)=\sigma.   
\end{equation}
Likewise, the reverse implication holds by noting that the dephasing operation is a thermal operation.
\end{proof}

Because the target state in our case is incoherent, we can use the above result to restate our problem as the existence of a thermal operation $\E$ approximately performing the following transformation 
\begin{equation}
    \D(\psi^{\otimes N} \otimes \ketbra{E^A_0}{E^A_0}) ~\xrightarrow{~\E~}~ \ketbra{\tilde{E}^{N}_k}{\tilde{E}^{N}_k} \otimes \ketbra{E^A_1}{E^A_1}.
\end{equation}
Since
\begin{equation}\label{eq: Dephased initial state}
    \D(\psi^{\otimes N} \otimes \ketbra{E^A_0}{E^A_0})=\D(\psi^{\otimes N}) \otimes \ketbra{E^A_0}{E^A_0},
\end{equation}
our problem further reduces to understanding the structure of the incoherent state $\D(\psi^{\otimes N})$. It is block-diagonal in the energy eigenbasis and can be diagonalised using thermal operations (since unitaries in a fixed energy subspace are free operations). After such a procedure, we end up with an incoherent state that is described by the probability distribution $\v{P}^N$ over the multi-index set $\v{k}$ 
\begin{equation}
    \label{eq:Multinomial_prob_vector}
    P^N_{\v{k}} =\binom{N}{k_1,...,k_d}\prod_{i=1}^d p^{k_i}_i.
\end{equation}
Note that $P^N_{\v{k}}$ specifies the probability of $k_1$ systems being in energy state $E_1$, $k_2$ systems being in energy state $E_2$, and so on; and that we made a technical assumption that energy levels are incommensurable, so that each vector $\v{k}$ corresponds to a different value of total energy.

We have thus reduced the problem of thermodynamic distillation from pure states to thermodynamic distillation from incoherent states. More precisely, let us denote the sharp distributions corresponding to $\ket{E^A_i}$ by $\v{s}^A_i$ and the corresponding flat states after embedding by $\v{f}^A_i$. As before, we also use $\tilde{\v{s}}_k^N$ and $\tilde{\v{f}}_k^N$ to denote distributions related to the sharp state $\ket{\tilde{E}^{N}_k}$ and its corresponding flat state. The embedded Gibbs state corresponding to $H^N$  will be again denoted by $\hat{\v{G}}^N$, however now it has an even simpler form than in Eq.~\eqref{eq:PN_GN}, as the initial systems have identical Hamiltonians:
\begin{equation}
    \hat{\v{G}}^N=\hat{\v{\gamma}}^{\otimes N} = {\v{\eta}}^{\otimes N}.
\end{equation}
Similarly, $\hat{\v{P}}^N$ will be used to denote the embedded initial state (even though it now has a different form than in Eq.~\eqref{eq:PN_GN}):
\begin{equation}
     \hat{P}^N_{\v{k}, g_{\v{k}}} = \binom{N}{k_1,...,k_d}\prod^{d}_{i=1}\frac{p^{k_i}_i}{D^{k_i}_i}, 
\end{equation}
where $\gamma_i=D_i/D$ and 
\begin{equation}
    g_{\v{k}} \in \bigg\{1, ..., \prod_{i=1}^d D^{k_i}_i \bigg\}
\end{equation}
is an index for the degeneracy coming from embedding. With the notation set, our distillation problem can now be written as
\begin{equation}
\label{thermomajorisation__embedded_pure_state}
    {\hat{\v{P}}}^N \otimes \v{f}^A_0 \otimes \tilde{\v{\eta}} \succ_{\epsilon}  \hat{\v{G}}^N \otimes \v{f}^A_1 \otimes \tilde{\v{f}_k}.
\end{equation}

% ------------------------------------------------
% SECTION V.D.2
% ------------------------------------------------

\subsubsection{Upper bound for the transformation error}

We begin by observing that our target distribution in Eq.~\eqref{thermomajorisation__embedded_pure_state} is flat, and so \mbox{$V(\hat{\v{G}}^N \otimes \v{f}^A_1 \otimes \tilde{\v{f}_k})=0$}. Thus, we can employ Lemma~\ref{lem:1shot-distill} and Eq.~\eqref{eq:simpleobservation2} to get the following expression for the optimal transformation error:
\begin{equation}
    \label{eq:Exactexpoferror}
    \epsilon_N= 1-\sum_{j=1}^{L}(\hat{\v{P}}^N)^{\downarrow}_j
\end{equation}
where $L$ is given by
\begin{align}
    L & = \exp[H(\hat{\v{G}}^{N})+H(\v{f}^A_1)+H(\tilde{\v{f}}_k)-H(\v{f}^A_0)-H(\tilde{\v{\eta}})]\nonumber\\
    & = \exp[H(\hat{\v{G}}^{N})-D(\tilde{\v{f}}_k\|\tilde{\v{\eta}})-\beta(E_1^A-E_0^A)]. \label{eq:L}
\end{align}
Notice that in the current case $\Delta F^N$, defined in Eq.~\eqref{eq:deltaF}, is given by
\begin{align}
   \!\!\! \Delta F^N&\!=\!\frac{1}{\beta}\Big(D(\psi^{\otimes N}\|\gamma^{\otimes N})+ D(\ketbra{E_0^A}{E_0^A}\|\gamma_A)\nonumber\\
    \!\!&\qquad\quad-D(\ketbra{\tilde{E}^N_{k}}{\tilde{E}^N_{k}}\|\tilde{{\gamma}}) -D(\ketbra{E_1^A}{E_1^A}\|\gamma_A)\Big)\nonumber\\
   \!\!\!  &\!=\! \frac{1}{\beta}\Big(ND(\psi\|\gamma)\!-\!D(\tilde{\v{f}}_k\|\tilde{\v{\eta}})\!-\!\beta(E_1^A\!-\!E_0^A)\Big).
\end{align}
Using the above we can then rewrite $L$ as
\begin{align}
    \label{eq:lnL}
    \log L =&  \beta \Delta F^N+H(\hat{\v{G}}^N)-ND(\psi\|\gamma).
\end{align}
Now, employing Eq.~\eqref{eq:lower} and the above, we provide the upper bound for $\epsilon_N$:
\begin{align}
 \!\!\!\epsilon_N&\leq 1- \sum_{\v{k},g_{\v{k}}}\left\lbrace 
    \hat{P}^N_{\v{k},g_{\v{k}}} \middle|
    \hat{P}^N_{\v{k},g_{\v{k}}}\geq \frac{1}{L}	\right\rbrace\nonumber\\
 \!\!\!   &= 1- \sum_{\v{k}}\left\lbrace 
    {P}^N_{\v{k}} \middle|
    {P}^N_{\v{k}}\geq \frac{\prod_{i=1}^d D_i^{k_i}}{L}	\right\rbrace\nonumber\\
    \!\!\!&= 1- \sum_{\v{k}}\left\lbrace 
    {P}^N_{\v{k}} \middle|
    \log{P}^N_{\v{k}}\geq \sum_{i=1}^d k_i\log D_i-\log L	\right\rbrace \nonumber\\
    \!\!\!&= 1- \sum_{\v{k}}\bigg\{
    {P}^N_{\v{k}} \Bigg|
    \frac{\log{P}^N_{\v{k}}}{N}\geq \sum_{i=1}^d \frac{k_i}{N}\log \gamma_i+D(\psi\|\gamma)\nonumber\\
    \!\!\!&\phantom{= 1- \sum_{\v{k}}\bigg\{
    {P}^N_{\v{k}} \Bigg|
    \log{P}^N_{\v{k}}\geq}-\frac{\beta}{N} \Delta F^N \Bigg\}.\label{eq:ErrorCalc}
 \end{align}

To simplify the calculation of the upper bound of $\epsilon_N$, we rewrite each $\v{k}$ as a function of a vector $\v{s}$ such that 
\begin{equation}
    \label{eq:k_of_s}
    \v{k}=\v{k}(\v{s})=N\v{p}+\sqrt{N}\v{s},
\end{equation}
with $\sum_{i=1}^d s_i=0$. We then note that
\begin{equation}\label{eq:KLD_of_psi}
    D(\psi\|\gamma)=-\sum_{i=1}^d p_i \log \gamma_i
\end{equation}
and so the condition in Eq.~\eqref{eq:ErrorCalc} can be rewritten as
\begin{align}\label{eq:logP^N}
    \frac{\log{P}^N_{\v{k}(\v{s})}}{N}&\geq \frac{1}{\sqrt{N}}\sum_{i=1}^d s_i\log \gamma_i-\frac{\beta}{N} \Delta F^N \nonumber\\
    & = -\frac{\beta}{\sqrt{N}}\sum_{i=1}^d s_iE_i-\frac{\beta}{N} \Delta F^N.
 \end{align}
As we rigorously argue in Appendix~\ref{app:log}, the left-hand side of the above vanishes much quicker than the right-hand side when $N\to\infty$, leading to
\begin{align}
 &\!\!\!\!\lim_{N\to\infty}\epsilon_N\nonumber\\
 &\!\!\!\!\quad\leq 1\!-\! \lim_{N\to\infty}\sum_{\v{s}}\bigg\{
    {P}^N_{\v{k}(\v{s})} \Bigg|
     \sum_{i=1}^d s_iE_i\geq- \frac{\Delta F^N}{\sqrt{N}}\Bigg\}.\!\label{Final_Bound_on_Epsilon}
\end{align}

\begin{figure*}[!htb]
    \centering
    \includegraphics{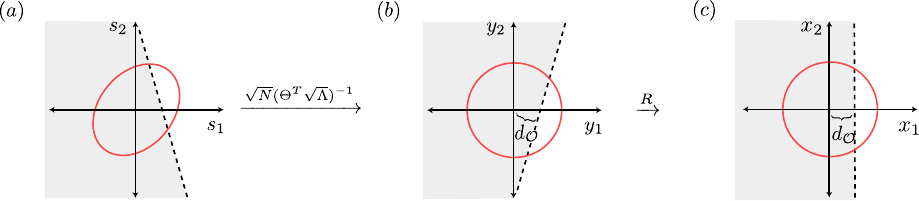}
    \caption{\label{fig:hyperplane} \textbf{Standardising the bivariate normal distribution}. The points with equal probability density for the bivariate normal distribution are represented by a red ellipsis centred at the origin, and the black dashed line corresponds to the constraining hyperplane. The upper bound on $\epsilon_N$ is given by the probability mass within the shaded area. In order to calculate it, we first apply a rotation and scaling transformation, making the ellipsis symmetric with respect to the origin. Then, using the rotational symmetry of the standard bivariate normal distribution, one can rotate it such that the hyperplane becomes parallel to $x_2$.} 
\end{figure*}

Our goal is then to calculate the sum of $P^N_{\v{k}(\v{s})}$ in the limit $N\to\infty$ subject to the following hyperplane constraint 
\begin{equation}\label{eq:Defining_Hyperplane}
     \v{s}\cdot\v{E}\geq- \frac{\Delta F^N}{\sqrt{N}},
\end{equation} 
where $\v{E}$ is a vector of energies (eigenvalues of $H$). First, we approximate the multinomial distribution $\v{P}^N$ specified in Eq.~\eqref{eq:Multinomial_prob_vector} by a multivariate normal distribution $\mathcal{N}^{(\v{\mu},\v{\Sigma})}$ with mean vector $\v{\mu}=N\v{p}$ and covariance matrix $\v{\Sigma} = N(\text{diag }(\v{p})-\v{p}\v{p}^{T})$:
\begin{align}
\mathcal{N}^{(\v{\mu},\v{\Sigma})}_{\v{k}(\v{s})}&= \frac{1}{\sqrt{(2\pi)^d|\boldsymbol\Sigma|}}
\exp\left(-\frac{1}{2}({\v{k}}-{\v{\mu}})^T{\boldsymbol\Sigma}^{-1}({\v{k}}-{\v{\mu}})
\right)\nonumber\\
&=\frac{1}{\sqrt{(2\pi)^d|\boldsymbol\Sigma|}}
\exp\left(-\frac{1}{2}\v{s}^T{N\boldsymbol\Sigma}^{-1}\v{s}
\right).
\end{align}
As we explain in Appendix~\ref{app:clt}, such an approximation can always be made with an error approaching 0 as $N\rightarrow\infty$. Next, we standardise the multivariate normal distribution $\mathcal{N}^{(\v{\mu},\v{\Sigma})}$ using rotation and scaling transformations:
\begin{equation}
\label{ratemodifiedNew2}
    \v{\Sigma} = \Theta^T\sqrt{\v{\Lambda}}\sqrt{\v{\Lambda}}\Theta,
\end{equation}
where $\v{\Lambda}$ is a diagonal matrix with the eigenvalues of $\v{\Sigma}$ and $\Theta$ is an orthogonal matrix with columns given by the eigenvectors of $\v{\Sigma}$. We illustrate this process for a three-level system (so described by $s_1$ and $s_2$ since $\sum_i s_i=0$) in Fig.~\ref{fig:hyperplane}. This rotation and scaling of co-ordinates allows us to write $\mathcal{N}^{(\v{\mu},\v{\Sigma})}$ as a product of univariate standard normal distribution $\phi(y_i)$:
\begin{align}\label{MultivariateNormalmean0}
\!\!\!\!\mathcal{N}^{(\v{\mu},\v{\Sigma})}_{\v{k}(\v{s}(\v{y}))}&\!=\!\frac{1}{\sqrt{(2\pi)^d|\boldsymbol\Sigma|}}
\exp\left(\!-\frac{1}{2}\v{y}^T\v{y}
\!\right)=\prod_{i=1}^d \phi(y_i),\!
\end{align}
where
\begin{equation}\label{y_and_s}
   \v{y}= \sqrt{N}(\Theta^T\sqrt{\v{\Lambda}})^{-1} \v{s}.
\end{equation}

We then can equivalently write the equation specifying the hyperplane, Eq.~\eqref{eq:Defining_Hyperplane}, as
\begin{equation}
    \label{transhyperplane}
    (\Theta^T\sqrt{\v{\Lambda}}\v{y})\cdot\v{E}\geq -\Delta F^N.
\end{equation}
Observe that the standard normal distribution given in  Eq.~\eqref{MultivariateNormalmean0} is rotationally invariant about the origin. One can thus choose a coordinate system \mbox{$\v{x}=\{x_1,\ldots, x_d\}$} by applying a suitable rotation $R$ on \mbox{$\v{y}=\{y_1,\ldots, y_d\}$}, so that the hyperplane specified in Eq.~\eqref{eq:Defining_Hyperplane} becomes parallel to all coordinate axes but  the $x_1$ axis. Eq.~\eqref{MultivariateNormalmean0} can then be rewritten in the following form
\begin{align}\label{MultivariateNormalmean0onx}
\!\!\!\!\mathcal{N}^{(\v{\mu},\v{\Sigma})}_{\v{k}(\v{s}(\v{x}))}&\!=\!\frac{1}{\sqrt{(2\pi)^d|\boldsymbol\Sigma|}}
\exp\left(\!-\frac{1}{2}\v{x}^T\v{x}
\!\right)=\prod_{i=1}^d \phi(x_i).\!
\end{align}
As we have 
\begin{equation}\label{Rotation}
    \v{x}=R\v{y},
\end{equation}
we can use it together with Eq.~\eqref{y_and_s} to rewrite Eq.~\eqref{eq:Defining_Hyperplane} as
\begin{equation}\label{eq:Defining_Hyperplane_New}
     \Theta^T\sqrt{\v{\Lambda}}R^T\v{x}\cdot\v{E}\geq -\Delta F^N.
\end{equation}

To calculate the right hand side of the inequality given in Eq.~\eqref{Final_Bound_on_Epsilon} in the limit $N\rightarrow\infty$, we integrate Eq. $\eqref{MultivariateNormalmean0onx}$ from $-\infty$ to $d_{\mathcal{O}}$ along $x_1$, and  from $-\infty$ to $+\infty$ along any other $x_i\neq x_1$, where $d_\mathcal{O}$ is the signed distance of the hyperplane given in Eq.~\eqref{eq:Defining_Hyperplane_New} from the origin (see Fig.~\ref{fig:hyperplane}). This distance can be explicitly calculated as
\begin{eqnarray}
 d_{\mathcal{O}} &=& \frac{\Delta F^N}{\sqrt{\v{E}\cdot(R\sqrt{\v{\Lambda}}\Theta)^T(R\sqrt{\v{\Lambda}}\Theta)\v{E}}} \nonumber\\
 &=& \frac{\Delta F^N}{\sqrt{\v{E}\cdot(\v{\Sigma} \v{E})}}
= \frac{\Delta F^N}{\sigma(F^N)},
\end{eqnarray}
where we have used the definition of $\sigma(F^N)$ from  Eq.~\eqref{eq:vn} in the last line. Thus, the upper bound on $\epsilon_N$ in the limit $N\rightarrow\infty$ given in Eq.~\eqref{Final_Bound_on_Epsilon} can be calculated as
\begin{align}
&1\!-\!\lim_{N\rightarrow\infty}\int_{-\infty}^{+\infty}\!dx_d\phi(x_d)\ldots\int_{-\infty}^{+\infty}\!dx_2\phi(x_2)\int_{-\infty}^{d_\mathcal{O}}\!dx_1\phi(x_1)\nonumber\\ 
&=  1\!-\!\lim_{N\rightarrow\infty}\Phi(d_\mathcal{O}) = \lim_{N\rightarrow\infty}\Phi\left(-\frac{\Delta F^N}{\sigma(F^N)}\right). \label{eq:Errorint}
\end{align}

% ------------------------------------------------
% SECTION V.D.3
% ------------------------------------------------

\subsubsection{Lower bound for the transformation error}\label{Lower_Bound_Error_transformation}

We start by writing $L$ from Eq.~\eqref{eq:lnL} as 
\begin{eqnarray}\label{Ldef_x}
    L= \exp\left(AN+x\sqrt{N v_N}\right)=:L(x),
\end{eqnarray}
where 
\begin{align}\label{A_and_x}
 A =&H(\v{\eta})-D(\psi\|{\gamma}),\qquad
 x=\frac{\Delta F^N}{\sigma(F^N)}.
\end{align}
In the previous section we have exactly calculated the right hand side of Eq.~\eqref{eq:ErrorCalc} in the limit $N\rightarrow\infty$ (see Eq.~\eqref{eq:Errorint}). Using Eqs.~\eqref{Ldef_x}-\eqref{A_and_x}, we can equivalently rewrite this as
\begin{equation}\label{eq:P_N_and_PHI}
     \lim_{N\rightarrow\infty}\sum_{\v{k},g_{\v{k}}}\left\lbrace 
    \hat{P}^N_{\v{k},g_{\v{k}}} \middle|
    \hat{P}^N_{\v{k},g_{\v{k}}}\geq \frac{1}{L(x)} \right\rbrace = \lim_{N\rightarrow\infty}\Phi(x),
\end{equation}
where $x$ depends on $N$ as per Eq.~\eqref{A_and_x}. Now, to prove the lower bound on transformation error $\epsilon_N$ we start with the exact expression, Eq.~\eqref{eq:Exactexpoferror}, and use the inequality from Eq.~\eqref{eq:upper}. Similarly as before, we choose \mbox{$\alpha=\exp(\delta\sqrt{N})$} such that $\delta>0$, in the Eq.~\eqref{eq:upper}. Thus from Eq.~\eqref{chi_order} and Eq.~\eqref{Ldef_x} we have,
\begin{equation}\label{eq:BoundonlxPure}
    L(x)\leq \frac{\chi_{\hat{\v{P}}^N}(e^{\delta\sqrt{N}}L(x))}{c}=\frac{\chi_{\hat{\v{P}}^N}(L(x+\delta))}{c}
\end{equation}
where $c$ can be evaluated similarly as before: 
\begin{align}\label{c_defn}
    c&=e^{\delta \sqrt{N}/2}  \sum_{i=\chi_{\hat{\v{P}}^N}(e^{\delta\sqrt{N}/2}L(x))}^{\chi_{\hat{\v{P}}^N}(e^{\delta\sqrt{N}}L(x))} (\hat{\v{P}}^N)_i^\downarrow\nonumber\\
    &=e^{\delta \sqrt{N}/2}  \sum_{i=\chi_{\hat{\v{P}}^N}(L(x+\delta/2))}^{\chi_{\hat{\v{P}}^N}(L(x+\delta))} (\hat{\v{P}}^N)_i^\downarrow \nonumber\\
    &=e^{\delta \sqrt{N}/2} \left( \sum_{i}\left\lbrace 
    \hat{P}^N_i \middle|
    \hat{P}_i^N \geq \frac{1}{L(x+\delta)}	
    \right\rbrace  \right.\nonumber\\
    &\quad\qquad\qquad\left. -\sum_{i}\left\lbrace 
    \hat{P}^N_i \middle|
    \hat{P}_i^N \geq \frac{1}{L(x+\delta/2)}
    \right\rbrace \right).
\end{align}
Using Eq.~\eqref{eq:P_N_and_PHI}, we see that the limiting behaviour of $c$ from Eq.~\eqref{c_defn} is given by 
\begin{equation}
  \lim_{N\to\infty} c = \left(\Phi(x+\delta)-\Phi(x+\delta/2)\right) \lim_{N\to\infty} e^{\delta \sqrt{N}/2}.
\end{equation}
Thus, for any finite $\delta>0$, there always exists $N_0$ such that for all $N\geq N_0$ we have $c>1$. Combining with Eq.~\eqref{eq:BoundonlxPure}, for large enough $N$ we finally have, 
\begin{equation}\label{eq:Boundonlx2Pure}
    L(x)\leq \frac{\chi_{\hat{\v{P}}^N}(L(x+\delta))}{c} \leq \chi_{\hat{\v{P}}^N}(L(x+\delta)).
\end{equation}
Hence, using Eq.~\eqref{eq:Boundonlx2Pure}, we have the following lower bound on transformation error
\begin{align}
    \label{eq:second_upper3}
   \!\!\! \lim_{N\rightarrow\infty}\epsilon_N &\geq  1- \lim_{N\rightarrow\infty}\sum_{i=1}^{\chi_{\hat{\v{P}}^N}(L(x+\delta))} (\hat{\v{P}}^N)_i^\downarrow \, \nonumber\\
   \!\!\! &=1- \lim_{N\rightarrow\infty}\sum_{i}\left\lbrace 
    \hat{P}^N_i \middle|
    \hat{P}_i^N \geq \frac{1}{L(x+\delta)}
    \right\rbrace \nonumber\\
    \!\!\!&= 1-\lim_{N\rightarrow\infty}\Phi(x+\delta)=\lim_{N\rightarrow\infty}\Phi(-x-\delta),
\end{align}
where the first equality in the last line follows from Eq.~\eqref{eq:P_N_and_PHI}.
 Since the above inequality holds for any $\delta>0$, taking the limit $\delta\rightarrow0$ we conclude that
\begin{equation}
    \lim_{N\to\infty} \epsilon_N\geq \lim_{N\to\infty} \Phi\left(-\frac{\Delta F^N}{\sigma(F^N)}\right).
\end{equation}
Finally, combining the above with the bound obtained in Eq.~\eqref{eq:Errorint}, we have
\begin{equation}\label{Proof_end_error}
    \lim_{N\to\infty}\epsilon_N=\lim_{N\to\infty} \Phi\left(-\frac{\Delta F^N}{\sigma(F^N)}\right),
\end{equation}
which completes the proof.

% ------------------------------------------------
% SECTION V.E
% ------------------------------------------------

\subsection{Proof of Theorem~\ref{thm:pure2}}
\label{sec:proof2b}

The proof of Theorem~\ref{thm:pure2} will be divided into two parts. First, we will find the embedded version of the optimal final state minimising the dissipation of free energy $F^N_{\mathrm{diss}}$, and derive the lower bound for $F^N_{\mathrm{diss}}$ as a function of the initial state. And then, we will calculate this bound up to second order asymptotic terms.

% ------------------------------------------------
% SECTION V.E.1
% ------------------------------------------------

\subsubsection{Deriving bound for optimal dissipation}

We start by applying Lemma~\ref{Opt_Diss_free} to the central Eq.~\eqref{thermomajorisation__embedded_pure_state} that specifies the conditions for the investigated $\epsilon$-approximate thermodynamic distillation process. As a result, the actual total final state in the embedded picture, $\hat{\v{F}}^N$, which is $\epsilon$ away from the target state $\hat{\v{G}}^N  \otimes \v{f}_1^A\otimes \tilde{\v{f}}_k$, is given, up to permutations, by
\begin{eqnarray}
\label{eq:distribution-total2}
     \hat{\v{F}}^N&=&
     \begin{cases}
       \frac{1-\epsilon}{k} &\quad \mathrm{for~} i\leq k,\\
       ( {\hat{\v{P}}}^N \otimes \v{f}^A_0 \otimes \tilde{\v{\eta}})_{i}^\downarrow &\quad\mathrm{for~}i>k,
     \end{cases}\, 
\end{eqnarray}
where
\begin{equation}
    k=\exp\Big(H(\hat{\v{G}}^N)+H(\v{f}^A_1)+H(\tilde{\v{f}}_k)\Big).% = D \tilde{D}.
\end{equation}
Similarly as before, we have chosen $\hat{\v{F}}^N$ with minimal entropy (i.e., the bistochastic map $B$ from Lemma~\ref{Opt_Diss_free} was chosen to be a permutation), as this minimises the dissipation of free energy. Due to the fact that the final state is specified only up to permutations, one may freely rearrange its elements to reduce the dissipation of free energy coming from unembedding (such dissipation happens when a state is not uniform within each embedding box). However, unlike in the previous case, this permutational freedom does not necessarily allow one to completely avoid additional dissipation. In the previous case, the number of non-uniform embedding boxes was exponentially small, which led to a negligible loss of free energy due to unembedding. But here, the number of non-uniform embedding boxes is not exponentially small and that can lead to a finite dissipation of free energy. Therefore, in this case, we can only provide a lower bound on dissipation of free energy as follows:
\begin{align}
\label{eq:Inequality_F_Diss}
     F^N_{\text{diss}} =&\frac{1}{\beta }\Big(D( \psi^{\otimes N}\|\v{G}^{ N})+ D(\ketbra{E_0^A}{E_0^A}\|\gamma_A)\nonumber\\
     &\qquad-D(\v{F}^N\|\v{G}^N\otimes\gamma_A\otimes\tilde{\v{G}}^N )\Big)\nonumber\\
     \geq& \frac{1}{\beta }\Big(D(\hat{\v{P}}^N\|\hat{\v{G}}^N)+ H(\v{P}^N) +D(\ketbra{E_0^A}{E_0^A}\|\gamma_A)\nonumber\\
     &\qquad- D(\hat{\v{F}}^N\|\hat{\v{G}}^N\otimes\hat{\gamma}_A\otimes\hat{\tilde{\v{G}}}^N)\Big),
\end{align}
where we have used the fact that
\begin{align}
    D(\psi^{\otimes N}\|\v{G}^N)&=D(\D(\psi^{\otimes N})\|\v{G}^N) + D(\psi^{\otimes N}\|\D(\psi^{\otimes N}))    \nonumber\\
    &=D(\hat{\v{P}}^N\|\hat{\v{G}}^N)+ H(\v{P}^N).
\end{align}

Next, we recall that $H(\v{P}^N)=O(\log N)$ (see Appendix~\ref{app:log} for details), and note that \mbox{$D(\ketbra{E_0^A}{E_0^A}\|\gamma_A) = -\log\gamma_0^A$}. Thus, the inequality from Eq.~\eqref{eq:Inequality_F_Diss} can be simplified as follows:
\begin{align}
     \!\!\!F^N_{\text{diss}} \gtrsim&\frac{1}{\beta}\Big(D(\hat{\v{P}}^N\|\hat{\v{G}}^N)
     -D(\hat{\v{F}}^N\|\hat{\v{G}}^N\otimes\hat{\gamma}_A\otimes\hat{\tilde{\v{G}}}^N)\nonumber\\
  &\qquad-\log\gamma_0^A\Big)\nonumber\\
  \!\!\!=&\frac{1}{\beta}\Big( H(\hat{\v{F}}^N)\!-\!H(\hat{\v{P}}^N)\!-\!\log (\tilde{D}D^A)\!-\!\log\gamma_0^A\Big).\label{simplified_True_Diss_Pure}
\end{align}
Here, $\tilde{D}$ and $D^A$ are the embedding constants defined by Eq.~\eqref{eq:definitiongibbsvec} for $\tilde{\v{G}}^N$ and $\gamma_A$, respectively, and similarly $D$ will denote this embedding constant for $\v{G}^N$. As before, $\tilde{D}$ and $D$ can be chosen to be equal, since the only thing that matters is that $\tilde{{G}}^N_k=\tilde{D}_k/\tilde{D}$ and ${{G}}^N_k={D}_k/D$, i.e., the change in $\tilde{D}$ or $D$ can be compensated by the appropriate change in $\tilde{D}_k$ or $D_k$. Furthermore, noting that $k=D\tilde{D}_kD^A_1 =\tilde{D}\tilde{D}_kD^A_1 $, we can express the entropy of $\hat{\v{F}}^N$ as
\begin{align}
    H(\hat{\v{F}}^N)=&-\sum_{i=1}^k \frac{1-\epsilon}{k}\log\left(\frac{1-\epsilon}{k}\right)\nonumber\\ 
    &- \sum_{i > k} ( {\hat{\v{P}}}^N \otimes \v{f}^A_0 \otimes \tilde{\v{\eta}})_i^{\downarrow} \log ( {\hat{\v{P}}}^N \otimes \v{f}^A_0 \otimes \tilde{\v{\eta}})_i^{\downarrow}\nonumber\\
    =&-(1-\epsilon)\log(1-\epsilon)+(1-\epsilon)\log k\nonumber\\
    &- \sum_{i > k} ( {\hat{\v{P}}}^N \otimes \v{f}^A_0 \otimes \tilde{\v{\eta}})_i^{\downarrow}\log \Big(\frac{{\hat{\v{P}}}^N}{\tilde{D}D_0^A}\Big)_i^{\downarrow}\nonumber\\
    \simeq&\log k-\epsilon\log L- \sum_{i > L} (\hat{\v{P}}^N)_i^{\downarrow}\log (\hat{\v{P}}^N)_i^{\downarrow} \label{eq:Entropy_of_J},
\end{align}
where in the last step of we used the fact that \mbox{$L=k/(\tilde{D}D_0^A)$}, which comes from Eq.~\eqref{eq:L}. Thus, from Eq.~\eqref{simplified_True_Diss_Pure}, the dissipated free energy in the optimal distillation process is simply bounded by
\begin{align}
    F^N_{\text{diss}}\gtrsim& \frac{1}{\beta}\Bigg(\sum_{i=1}^{L} (\hat{\v{P}}^N)_i^{\downarrow}\log (\hat{\v{P}}^N)_i^{\downarrow}+\log\Big(\frac{k}{\tilde{D}D^A}\Big)\nonumber\\
    &\qquad-\epsilon\log L- \log \gamma_0^A\Bigg)\nonumber\\
   =& \frac{1}{\beta}\Bigg(\!\sum_{i=1}^{L} (\hat{\v{P}}^N)_i^{\downarrow}\log (\hat{\v{P}}^N)_i^{\downarrow}+(1-\epsilon)\log L\!\Bigg).
\end{align}
Finally, using the expression for $L$ given in Eq.~\eqref{Ldef_x}, we can rewrite the above bound as 
\begin{align}
    F^N_{\text{diss}}\gtrsim& \frac{1}{\beta}\Bigg(\!\sum_{i=1}^{L} (\hat{\v{P}}^N)_i^{\downarrow}\log (\hat{\v{P}}^N)_i^{\downarrow}\nonumber\\
    &\qquad +(1-\epsilon)\left(AN+x\sqrt{Nv_N}\right)\!\Bigg),\label{eq: F_Diss_1}
\end{align}
where $A$ and $x$ are given by Eq.~\eqref{A_and_x}. 

% ------------------------------------------------
% SECTION V.E.2
% ------------------------------------------------

\subsubsection{Calculating bound for optimal dissipation}

We now proceed to bounding the first term on the right hand side of Eq.~\eqref{eq: F_Diss_1}. We start by noting that Eq.~\eqref{eq:Boundonlx2Pure} implies that the following holds for any $\delta>0$:   
\begin{eqnarray}\label{eq:ent_Bound2}
     \!\!\!\!\!\!\!\!\sum_{i=1}^{L(x)} (\hat{\v{P}}^N)_i^\downarrow\log(\hat{\v{P}}^N)_i^\downarrow\geq 
     \!\!\!\!\!\!\sum_{i=1}^{\chi_{\hat{\v{P}}^N}(L(x+\delta))}\!\!\!\!(\hat{\v{P}}^N)_i^\downarrow\log(\hat{\v{P}}^N)_i^\downarrow.
\end{eqnarray}
Next, from the definition of $\chi_{\hat{\v{P}}^N}(L(x+\delta))$, it follows that
\begin{align}
    \label{eq:Entropy_Calc_2}
    &\!\!\!\!\!\sum_{i=1}^{\chi_{\hat{\v{P}}^N}(L(x+\delta))} (\hat{\v{P}}^N)_i^\downarrow\log(\hat{\v{P}}^N)_i^\downarrow\nonumber\\
    &\quad=\sum_{i}\Big\{
    \hat{P}^N_i\log (\hat{P}^N_i) ~\Big|~
    \hat{P}^N_i \geq \frac{1}{L(x+\delta)}\Big\}\nonumber\\
    &\quad= \sum_{\v{k},g_{\v{k}}}\left\lbrace 
    \hat{P}^N_{\v{k},g_{\v{k}}}\log(\hat{P}^N_{\v{k},g_{\v{k}}}) ~\middle|~
    \hat{P}^N_{\v{k},g_{\v{k}}}\geq \frac{1}{L(x+\delta)}	\right\rbrace\nonumber \\
    &\quad=\sum_{\v{k}}\left\lbrace 
    {P}^N_{\v{k}}\log (\hat{P}^N_{\v{k},1}) ~\middle|~
    {P}^N_{\v{k}}\geq \frac{\prod_{i=1}^d D_i^{k_i}}{L(x+\delta)}	\right\rbrace\nonumber\\
    &\quad=\sum_{\v{k}}\left\lbrace 
    {P}^N_{\v{k}}\log ({P}^N_{\v{k}}) ~\middle|~
    {P}^N_{\v{k}}\geq \frac{\prod_{i=1}^d D_i^{k_i}}{L(x+\delta)}	\right\rbrace\nonumber\\
    &\qquad- \sum_{\v{k}}\left\lbrace 
    {P}^N_{\v{k}}\Big(\sum_{j=1}^d k_j\log\gamma_j\Big) ~\middle|~
    {P}^N_{\v{k}}\geq \frac{\prod_{i=1}^d D_i^{k_i}}{L(x+\delta)}	\right\rbrace\nonumber\\
    &\qquad-(1-\epsilon)N\log D.
\end{align}

We now note that the constraint on the summand in Eq.~\eqref{eq:Entropy_Calc_2} can be modified using the parametrization of $\v{k}$ as a function of $\v{s}$ given in Eq.~\eqref{eq:k_of_s} as follows: 
\begin{align}
    &\!\!\!\!\!\!\!\!\!\!\!\!\!\!\log (P^N_{\v{k}}) \nonumber\\
    ~~\geq& \sum_{i=1}^d k_i\log D_i - \log (L(x+\delta))\nonumber\\
    ~~=& \sum_{i=1}^d (Np_i+\sqrt{N}s_i)\log \gamma_i+N\log D- \log L(x+\delta)\nonumber\\
    ~~=& -\beta\sqrt{N}\sum_{i=1}^d s_iE_i+N\Big(\log D-D(\psi\|\gamma)\Big)\nonumber\\
    &- N\Big(\log D-D(\psi\|\gamma)\Big)-(x+\delta)\sqrt{Nv_N}\nonumber\\
    ~~=& -\beta\sqrt{N}\sum_{i=1}^d s_iE_i-(x+\delta)\sqrt{Nv_N}\nonumber\\
    ~~=& -\beta\left(\sqrt{N}\v{s}\cdot\v{E}+\delta\sigma(F^N)+\Delta F^N\right),
    \label{eq:deltabound}
\end{align}
where in the second equality we used the form of $L$ given in Eq.~\eqref{Ldef_x} and employed Eq.~\eqref{eq:KLD_of_psi}, whereas in the final equality we used the definition $\sigma(F^N)=\sqrt{Nv_n}/\beta$ and Eq.~\eqref{A_and_x} saying that \mbox{$x=\Delta F^N/\sigma(F^N)$}. Moreover, since \mbox{$\log(P^N_{\v{k}(\v{s})})=O(\log N)$}, the condition from Eq.~\eqref{eq:deltabound} can be rewritten as
\begin{equation}
    \label{eq:constraint2}
    \sqrt{N}\v{s}\cdot\v{E}+\Delta F^N+\delta\sigma(F^N)\gtrsim 0.
\end{equation}

Coming back, we can now rewrite Eq.~\eqref{eq:Entropy_Calc_2} using the parametrization of $\v{k}$ as a function of $\v{s}$ from Eq.~\eqref{eq:k_of_s} and re-expressing the constraint using Eq.~\eqref{eq:constraint2}:

\begin{align}
    \label{eq:Entropy_Calc_3}
    &\!\!\!\!\!\sum_{i=1}^{\chi_{\hat{\v{P}}^N}(L(x+\delta))} (\hat{\v{P}}^N)_i^\downarrow\log(\hat{\v{P}}^N)_i^\downarrow\nonumber\\
    &\quad=\sum_{\v{k}(\v{s})}\Bigg\{ 
    {P}^N_{\v{k}(\v{s})}\log ({P}^N_{\v{k}(\v{s})})-\sqrt{N}{P}^N_{\v{k}(\v{s})}\sum_{i=1}^d s_i\log\gamma_i \nonumber\\
    &\qquad\qquad\qquad\Bigg|    \sqrt{N}\v{s}\cdot\v{E}+\Delta F^N+\delta\sigma(F^N)\geq 0\Bigg\}\nonumber\\
    &\qquad\qquad-(1-\epsilon)N\Big(\log D-D(\psi\|\gamma)\Big)\nonumber\\
    &\quad\simeq\beta\sqrt{N}\sum_{\v{s}}\Bigg\{ 
    {P}^N_{\v{k}(\v{s})} \v{s}\cdot\v{E}\Bigg| \sqrt{N}\v{s}\cdot\v{E}+\Delta F^N\nonumber\\
    &\qquad\qquad\qquad\quad +\delta\sigma(F^N)\geq 0\Bigg\}-(1-\epsilon)AN,
\end{align}
where we used the definition of $A$ from Eq.~\eqref{A_and_x}. 

Our next goal goal is then to calculate the sum
\begin{equation}
    \label{eq:Defining_Hyperplane2}
    \sum_{\v{s}}\Big\{P^N_{\v{k}(\v{s})}\v{s}\cdot\v{E}  \bigg|\sqrt{N}\v{s}\cdot\v{E}+\Delta F^N+\delta\sigma(F^N)\geq 0\Big\}.
\end{equation}
Similarly as before, we will approximate the multinomial distribution $\v{P}^N$  by a multivariate normal distribution $\mathcal{N}^{(\v{\mu},\v{\Sigma})}$ with mean vector $\v{\mu}=N\v{p}$ and covariance matrix $\v{\Sigma} = N(\text{diag }(\v{p})-\v{p}\v{p}^{T})$:
\begin{align}
\mathcal{N}^{(\v{\mu},\v{\Sigma})}_{\v{k}(\v{s})}&= \frac{1}{\sqrt{(2\pi)^d|\boldsymbol\Sigma|}}
\exp\left(-\frac{1}{2}({\v{k}}-{\v{\mu}})^T{\boldsymbol\Sigma}^{-1}({\v{k}}-{\v{\mu}})
\right)\nonumber\\
&=\frac{1}{\sqrt{(2\pi)^d|\boldsymbol\Sigma|}}
\exp\left(-\frac{1}{2}\v{s}^T{N\boldsymbol\Sigma}^{-1}\v{s}
\right).
\end{align}
Next, we standardise the multivariate normal distribution $\mathcal{N}^{(\v{\mu},\v{\Sigma})}$ using rotation and scaling transformations:
\begin{equation}
\label{ratemodifiedNew}
    \v{\Sigma} = \Theta^T\sqrt{\v{\Lambda}}\sqrt{\v{\Lambda}}\Theta,
\end{equation}
where $\v{\Lambda}$ is a diagonal matrix with the eigenvalues of $\v{\Sigma}$ and $\Theta$ is an orthogonal matrix with columns given by the eigenvectors of $\v{\Sigma}$. This rotation and scaling of co-ordinates allows us to write $\mathcal{N}^{(\v{\mu},\v{\Sigma})}$ as a product of univariate standard normal distribution $\phi(y_i)$:
\begin{align}
\!\!\!\!\mathcal{N}^{(\v{\mu},\v{\Sigma})}_{\v{k}(\v{s}(\v{y}))}&\!=\!\frac{1}{\sqrt{(2\pi)^d|\boldsymbol\Sigma|}}
\exp\left(\!-\frac{1}{2}\v{y}^T\v{y}
\!\right)=\prod_{i=1}^d \phi(y_i),\!
\end{align}
where
\begin{equation}\label{eq:Def_y}
 \v{y} = \sqrt{N}(\Theta^{T}\sqrt{\Lambda})^{-1} \v{s},   
\end{equation}
such that
\begin{equation}\label{eq:Def_s}
\v{s} = \frac{1}{\sqrt{N}} \Theta^{T} \sqrt{\Lambda} \v{y}.
\end{equation}

From Eq.~\eqref{eq:Def_y} and Eq.~\eqref{eq:Def_s}, one can write equivalently the sum given in Eq.~\eqref{eq:Defining_Hyperplane2} as
\begin{equation}
    \label{eq:Defining_Hyperplane3}
    \sum_{\v{y}}\Big\{P^N_{\v{k}(\v{s}(\v{y}))}\frac{\tilde{\v{E}}\cdot \v{y}}{\sqrt{N}} \bigg|\tilde{\v{E}}\cdot\v{y}+\Delta F^N+\delta\sigma(F^N) \geq 0\Big\},
\end{equation}
where we defined $\tilde{\v{E}}:= \sqrt{\Lambda} \Theta \v{E}$ with the following normalisation
\begin{align}
\begin{split}
    \| \tilde{\v{E}}\| &= \sqrt{\tilde{\v{E}}\cdot\tilde{\v{E}}} = \sqrt{(\sqrt{\Lambda} \Theta \v{E})\cdot(\sqrt{\Lambda} \Theta \v{E})} \\ &= \sqrt{\v{E}\cdot\Theta^{T}\Lambda \Theta \v{E}} = \sqrt{\v{E}\cdot \v{\Sigma} \v{E}} \\ &=  \sigma( F^N).
\end{split}
\end{align}
One can then equivalently express Eq.~\eqref{eq:Defining_Hyperplane3} as 
\begin{align}
    \label{eq:sigma_times_P^N}
    \frac{\sigma(F^N)}{\sqrt{N}}\sum_{\v{y}}\Big\{P^N_{\v{k}(\v{s}(\v{y}))}\hat{\tilde{\v{E}}}\cdot \v{y} \;\bigg|\;&\sigma(F^N)\hat{\tilde{\v{E}}}\cdot\v{y}+\Delta F^N\nonumber\\
    &+\delta\sigma(F^N) \geq 0\Big\},
\end{align}
where $\hat{\tilde{\v{E}}}$ is unit vector along $\tilde{\v{E}}$. Now, we choose a rotation $R$ such that
\begin{equation}
    R\hat{\tilde{\v{E}}} = (1,\ldots, 0)^T
\end{equation}
and we call $R\v{y}=\v{x}$. Thus, we can rewrite Eq.~\eqref{eq:sigma_times_P^N} as
\begin{align}
    \label{eq:sigma_times_P^N2}
    &\frac{\sigma(F^N)}{\sqrt{N}}\sum_{\v{y}}\Big\{P^N_{\v{k}(\v{s}(\v{y}))}R\hat{\tilde{\v{E}}}\cdot R\v{y} \;\Bigg|\;\sigma(F^N)R\hat{\tilde{\v{E}}}\cdot R\v{y}\nonumber\\
    &\qquad\qquad\qquad+\Delta F^N+\delta\sigma(F^N) \geq 0\Big\}\nonumber\\
    &\quad=\frac{\sigma(F^N)}{\sqrt{N}}\sum_{\v{x}}\Big\{P^N_{\v{k}(\v{s}(\v{x}))}x_1 \;\Bigg|\;\sigma(F^N)x_1\nonumber\\
    &\qquad\qquad\qquad\qquad+\Delta F^N+\delta\sigma(F^N) \geq 0\Big\}.
\end{align}

Now, employing multivariate normal distribution approximation and the fact that it is isotropic, we can exactly calculate the above expression as an integral over $\v{x}$ of the following form 
\begin{align}
      &\frac{\sigma(F^N)}{\sqrt{N}}\Bigg(\int_{-\infty}^{+\infty}dx_2 \phi(x_2)\ldots \int_{-\infty}^{+\infty}dx_d \phi(x_d)\nonumber\\
      &\times\int_{-\infty}^{d_\mathcal{O}} dx_1x_1 \phi(x_1)\Bigg)=\frac{\sigma( F^N)}{\sqrt{2\pi N}}\exp\left(-\frac{d_{\mathcal{O}}^2}{2}\right),
\end{align}
where $d_\mathcal{O}$ is the distance between the origin and the hyperplane specified by the constraint given in Eq.~\eqref{eq:sigma_times_P^N2}:
\begin{equation}
    d_{\mathcal{O}} = \frac{ \Delta F^N+\delta\sigma(F^N) }{\sqrt{\tilde{\v{E}}\cdot\tilde{\v{E}}}}=\frac{ \Delta F^N}{\sigma( F^N)}+\delta.
\end{equation}
Thus, finally we get the following expression for the desired sum from Eq.~\eqref{eq:Defining_Hyperplane2}:
\begin{align}
    &\sum_{\v{s}}\Big\{P^N_{\v{k}(\v{s})}\v{s}\cdot\v{E}  \bigg|\sqrt{N}\v{s}\cdot\v{E}+\Delta F^N+\delta\sigma(F^N)\geq 0\Big\}\nonumber\\
    &\qquad=  \frac{\sigma( F^N)}{\sqrt{2\pi N}}\exp\left(-\frac{1}{2}\left(\frac{ \Delta F^N}{\sigma( F^N)}+\delta\right)^2\right).
\end{align}

Substituting the above to Eq.~\eqref{eq:Entropy_Calc_3}, and the resulting expression to Eq.~\eqref{eq:ent_Bound2}, we get the following bound:
\begin{align}
    &\!\!\!\!\sum_{i=1}^{L(x)} (\hat{\v{P}}^N)_i^\downarrow\log(\hat{\v{P}}^N)_i^\downarrow  \gtrsim -(1-\epsilon)A \nonumber\\
    &\qquad+\frac{\beta\sigma( F^N)}{\sqrt{2\pi }}\exp\left(-\frac{1}{2}\left(\frac{ \Delta F^N}{\sigma( F^N)}+\delta\right)^2\right),
\end{align}
which in turn can be used in Eq.~\eqref{eq: F_Diss_1} to give the following:
\begin{align}
    F^N_{\mathrm{diss}} \gtrsim &\frac{\sigma( F^N)}{\sqrt{2\pi }}\exp\left(-\frac{1}{2}\left(\frac{ \Delta F^N}{\sigma( F^N)}+\delta\right)^2\right)   \nonumber\\
    &+(1-\epsilon)\Delta F^N.
\end{align}
Since the above holds for any $\delta>0$, by taking the limit $\delta\rightarrow 0$, we finally arrive at
\begin{equation}\label{eq:F_N_Diss_Simp}
    \!\!\! F^N_{\rm diss} \gtrsim   (1-\epsilon)\Delta F^N +\frac{\sigma( F^N)}{\sqrt{2\pi}}\exp\Bigg({-\frac{(\Delta F^N)^2}{2\sigma( F^N)^2}}\Bigg).
\end{equation}
Employing Eq.~\eqref{Proof_end_error} we have
\begin{equation}
    \Delta F^N = -\Phi^{-1}(\epsilon)\sigma(F^N),
\end{equation}
and so substituting this in Eq.~\eqref{eq:F_N_Diss_Simp} we obtain
\begin{eqnarray}
F^N_{\rm diss}&\gtrsim& a(\epsilon)\sigma( F^N),
\end{eqnarray}
with $a$ given by Eq.~\eqref{eq:a}, which completes the proof.

% ------------------------------------------------
% SECTION VI - OUTLOOK
% ------------------------------------------------

\section{Outlook}
\label{sec:out}

In this paper we have derived a version of the fluctuation-dissipation theorem for state interconversion under thermal operations. We achieved this by establishing a relation between optimal transformation error and the amount of free energy dissipated in the process on the one hand, and fluctuations of free energy in the initial state of the system on the other hand. We addressed and solved the problem in two different regimes: for initial states being either energy-incoherent or pure, with the target state in both cases being an energy eigenstate, and with the possibility to change the Hamiltonian in the process. For the case of finitely many independent but not necessarily identical energy-incoherent systems, we have provided the single-shot upper bound on the optimal transformation error as a function of average dissipated free energy and free energy fluctuations. Moreover, in the asymptotic regime we obtained the optimal transformation error up to second order asymptotic corrections, which extends previous results of Ref.~\cite{Chubb2018beyondthermodynamic} to the regime of non-identical initial systems and varying Hamiltonians. For the first time we have also performed the asymptotic analysis of the thermodynamic distillation process from quantum states that have coherence in the energy eigenbasis. As a result, we expressed the optimal transformation error from identical pure states and free energy dissipated during this transformation up to second order asymptotic corrections as a function of free energy fluctuations.

Our work can be naturally extended in the following directions. Firstly, one could generalise our analysis to arbitrary initial states. We indeed believe that an analogous result to ours will hold for such general mixed states with coherence. That is because dephasing into fixed energy subspaces leads to free energy change of the order $O(\log N)$, which is negligible compared to the second order asymptotic corrections of the order $O(\sqrt{N})$ that we focus on. In other words, the contribution of coherence to free energy per copy of the system vanishes faster with growing $N$ than what we are interested while studying second order corrections. 

Secondly, it would be extremely interesting to generalise the thermodynamic state interconversion problem to arbitrary final states, and see how the interplay between the fluctuations of the initial and target states affects dissipation. For energy-incoherent initial and final states one can infer from Ref.~\cite{korzekwa2019avoiding} that appropriately tuned fluctuations can significantly reduce dissipation, however nothing is known for states with coherence. Unfortunately, since thermal operations are time-translation covariant, such that coherence and athermality form independent resources~\cite{lostaglio2015quantum, lostaglio2015description, PRLHorodeckicoherence}, it seems unlikely that the current approach can be easily generalised. Thirdly, one could try to extend our results on pure states to allow for non-identical systems and to derive a bound working for all $N$, not only for $N\to\infty$ (i.e., replace the proving technique based on central limit theorem by the one based on a version of Berry-Esseen theorem).

In our work we have also provided a number of physical applications of our fluctuation-dissipation theorems by considering several scenarios and explaining how our results can be useful to describe fundamental and well-known thermodynamic and information-theoretic processes. We derived the optimal value of extractable work in a thermodynamic distillation process as a function of the transformation error associated to the work quality. This, together with the knowledge of the actual final free energy of the battery system provided by Theorem~\ref{thm:incoherent2}, could potentially be used to clarify the notion of imperfect work~\cite{Aberg2013, Woods2019maximumefficiencyof, Ying_Ng_NJP}, and to construct a comparison platform allowing one to continuously distinguish between work-like and heat-like forms of energy. We have also shown how our results yield the optimal trade-off between the work invested in erasing $N$ independent bits prepared in arbitrary states, and the erasure quality measured by the infidelity distance between the final state and the fully erased state. This can of course be straightforwardly extended to higher-dimensional systems and arbitrary final erased state (not necessarily the ground state). Finally, we have investigated the optimal encoding rate into a collection of non-interacting subsystems consisting of energy-incoherent or pure states using thermal operations. We derived the optimal rate (up to second-order asymptotics) of encoding information with a given average decoding error and without spending thermodynamic resources. This provides an operational interpretation of the resourcefulness of athermal quantum states for communication scenarios under the restriction of using thermal operations.

We would also like to point out to some possible technical extensions of our results. Firstly, we used infidelity as our quantifier of transformation error, but we expect that similar results could be derived using other quantifiers, e.g., the trace distance. Secondly, our investigations were performed in the spirit of small-deviation analysis (where we look for constant transformation error and total free energy dissipation of the order $O(\sqrt{N})$), but possibly other interesting fluctuation-dissipation relations could be derived within the the moderate and large deviation regimes. Thirdly, our result for pure states is limited to Hamiltonians with incommensurable spectrum, but we believe this is just a technical nuisance that one should be able to get rid of. Lastly, within the framework of general resource theories, it might be possible to derive analogous fluctuation-dissipation relations, but with free energy replaced by a resource quantifier relevant for a given resource theory.

% ------------------------------------------------
% ACKNOWLEDGEMENTS
% ------------------------------------------------

\subsection*{Acknowledgements} 

KK would like to thank Chris Chubb and Marco Tomamichel for helpful discussions. KK and AOJ acknowledge financial support by the Foundation for Polish Science through TEAM-NET project (contract no. POIR.04.04.00-00-17C1/18-00). TB and MH acknowledge support from the Foundation for Polish Science through IRAP project co-financed by EU within the Smart Growth Operational Programme (contract no.2018/MAB/5).

\appendix

% ------------------------------------------------
% APPENDIX A
% ------------------------------------------------

\section{Optimality of the communication rate}
\label{app:optimality}

The following derivation will closely follow the proof of Lemma~1 of Ref.~\cite{korzekwa2019encoding}. Let us assume that for a system $(\rho^N,H^N)$ it is possible to encode $M$ messages in a thermodynamically-free way so that the average decoding error is $\epsilon$. It means that there exists a set of $M$ encoding thermal operations $\{\E_i\}_{i=1}^M$ and a decoding POVM $\{ \Pi_i \}_{i=1}^M$ such that
\begin{equation}
    \label{eq:assumption}
    1-\epsilon=\frac{1}{M}\sum_{i=1}^M \tr{\E_i(\rho^N) \Pi_i}.
\end{equation}
Note that every thermal operation $\E_i$ between the initial system $(\rho^N,H^N)$ and a target system $(\tilde{\rho}^N,\tilde{H}^N)$ preserves the thermal equilibrium state, i.e.,
\begin{equation}
    \E_i(\gamma^N)=\tilde{\gamma}^N.
\end{equation}
Now, let us introduce the following three states
\begin{subequations}
\begin{align}
    \tau:=&\frac{1}{M} \sum_{i=1}^M \ketbra{i}{i}\otimes \E_i(\rho^N),\\
    \zeta:=&\frac{1}{M} \sum_{i=1}^M \ketbra{i}{i}\otimes \gamma^N,\\
    \tilde{\zeta}:=&\frac{1}{M} \sum_{i=1}^M \ketbra{i}{i}\otimes \tilde{\gamma}^N.
\end{align}
\end{subequations}
The hypothesis testing relative entropy $D_H^{\epsilon}$ between $\tau$ and $\tilde{\zeta}$, defined by~\cite{wang10,buscemi2010quantum,brandao2011one} 
\begin{align}
\!\!\! D_H^{\epsilon}(\tau\|\tilde{\zeta}) \!:= - \log \inf \big\{ \tr{Q\tilde{\zeta}}\ \!\big|\ &0 \leq Q \leq \iden,\nonumber\\
\!\!\!\!\! &\tr{Q\tau} \geq 1-\epsilon \big\} ,\!
\label{eq:hypothesis_rel_ent}
\end{align}
satisfies the following
\begin{equation}
	D_H^{\epsilon}(\tau\|\tilde{\zeta})\geq -\log \tr{Q\tilde{\zeta}}
\end{equation}
for
\begin{equation}
	Q = \sum_{i=1}^M \ketbra{i}{i} \otimes \Pi_i.
\end{equation}
This is because the above (potentially suboptimal) choice of $Q$ clearly satisfies \mbox{$0\leq Q\leq \iden$} and also
\begin{equation}
	\tr{Q\tau}=	\frac{1}{M} \sum_{i=1}^M \tr{\E_i(\rho)\Pi_i}\geq 1-\epsilon,
\end{equation} 
due to our assumption in Eq.~\eqref{eq:assumption}. At the same time we have
\begin{equation}
	\tr{Q\tilde{\zeta}}=\frac{1}{M} \sum_{i=1}^M \tr{\tilde{\rho}^N \Pi_i}=\frac{1}{M},
\end{equation} 
so that 
\begin{equation}
	\label{eq:hypothesis_ineq_1}
	\log M\leq D_H^{\epsilon}(\tau\|\tilde{\zeta}).
\end{equation}

Next, we introduce the following encoding channel
\begin{equation}
	\F:=\sum_{i=1}^{M} \ketbra{i}{i} \otimes \E_i,
\end{equation}
which satisfies
\begin{equation}
    \F(\zeta)=\tilde{\zeta}.
\end{equation}
Employing the data-processing inequality twice, we get the following sequence of inequalities:
\begin{align}
	D_H^{\epsilon}(\tau\|\tilde{\zeta})&=D_H^{\epsilon}\left(\F\left(\frac{1}{M}\sum_{i=1}^M \ketbra{i}{i} \otimes \rho^N\right)\bigg\|\F({\zeta})\right)\nonumber\\
	&\leq D_H^{\epsilon}\left(\frac{1}{M}\sum_{i=1}^M \ketbra{i}{i} \otimes \rho^N\bigg\|\zeta\right)\nonumber\\
	&\leq D_H^{\epsilon}\big( \rho^N \big\| \gamma^N \big).
\end{align}
Combining this with Eq.~\eqref{eq:hypothesis_ineq_1}, we arrive at
\begin{equation}
	\label{eq:hypothesis_ineq_2}
	\log M\leq D_H^{\epsilon}\big( \rho^N \big\| \gamma^N \big).
\end{equation}

Finally, for the case of identical initial subsystems, $\rho^N=\rho^{\otimes N}$ and $\gamma^N=\gamma^{\otimes N}$, we can use the known second order asymptotic expansion of the hypothesis testing relative entropy~\cite{li12},
\begin{align}
    \frac{1}{N}D^\epsilon_H(\rho^{\otimes N}\|\gamma^{\otimes N})&\simeq D(\rho\|\gamma)+\sqrt{\frac{V(\rho\|\gamma)}{N}}\Phi^{-1}(\epsilon),
\end{align}
leading to
\begin{align}
    \frac{\log M}{N}\lesssim D(\rho\|\gamma)+\sqrt{\frac{V(\rho\|\gamma)}{N}}\Phi^{-1}(\epsilon).
\end{align}
For the above proof to work also in the case of non-identical subsystems, one would need to prove the following asymptotic behaviour of the hypothesis testing relative entropy:
\begin{align}
    \!&\!\frac{1}{N}D^\epsilon_H\left(\bigotimes_{n=1}^N \rho_n\bigg\|\bigotimes_{n=1}^N \gamma_n\right) \nonumber\\ \!&\!~\simeq\frac{1}{N}\sum\limits_{n=1}^N D(\rho_n\|\gamma_n)+\sqrt{\frac{\frac{1}{N}\!\sum\limits_{n=1}^N \!V(\rho_n\|\gamma_n)}{N}}\Phi^{-1}(\epsilon).\!
\end{align}

% ------------------------------------------------
% APPENDIX B
% ------------------------------------------------

\section{Proof of Lemma~\ref{Opt_Diss_free}}
\label{app:min_out}
\begin{proof}

Consider $\tilde{\v{q}}$ to be a probability vector that saturates $\v{p}\succ_{\epsilon}\v{q}$. By definition it means that $\v{p}\succ\tilde{\v{q}}$ and \mbox{$F(\v{q},\tilde{\v{q}})=1-\epsilon$}. Since $V(\v{q})=0$ and $H(\v{q})=\log L$, the probability vector $\v{q}$ contains exactly $L$ non-zero entries such that each of them is equal to $1/L$. Thus,
\begin{equation}
    \v{q}^{\downarrow}=\Big(\underbrace{\frac{1}{L},\ldots,\frac{1}{L}}_{L},0,\ldots,0\Big).
\end{equation}
From the definition of $F(\v{q},\tilde{\v{q}})$ given in Eq.~\eqref{eq:DefFidelity} we have the following
\begin{eqnarray}
\label{eq:Fidelity_majorize}
1-\epsilon=F(\v{q},\tilde{\v{q}}) &\leq& F(\v{q}^{\downarrow},\tilde{\v{q}}^{\downarrow})\nonumber\\
&=& \frac{1}{L}\left(\sum_{i=1}^L \sqrt{\tilde{q}_i^{\downarrow}}\right)^2= \frac{1}{L}\left(\sum_{i=1}^L \sqrt{\tilde{q}_i^{\downarrow}\cdot 1}\right)^2\nonumber\\&\leq& \frac{1}{L}\left(\sum_{i=1}^L\tilde{q}_i^{\downarrow}\right)  \left(\sum_{i=1}^L 1\right) = \sum_{i=1}^L\tilde{q}_i^{\downarrow},
\end{eqnarray}
where the first inequality follows from the definition of fidelity, while the second one comes from the Cauchy-Schwarz inequality. Now, from Lemma~\ref{Opt_Diss_free}, we have 
\begin{eqnarray}
\sum_{i=1}^L p_i^{\downarrow} = 1-\epsilon.
\end{eqnarray}
Combining this with Eq.~\eqref{eq:Fidelity_majorize} we obtain
\begin{eqnarray}
\sum_{i=1}^L p_i^{\downarrow}\leq \sum_{i=1}^L \tilde{q}_i^{\downarrow}.
\end{eqnarray}
On the other hand, $\v{p}\succ \tilde{\v{q}}$ gives 
\begin{equation}
    \sum_{i=1}^L p_i^{\downarrow}\geq \sum_{i=1}^L \tilde{q}_i^{\downarrow},    
\end{equation}
and so we conclude that
\begin{equation}\label{constrainoflagrangemultiplier}
    \sum_{i=1}^L p_i^{\downarrow}= \sum_{i=1}^L \tilde{q}_i^{\downarrow}=1-\epsilon.
\end{equation}

Next, note that as $\epsilon = 1-F(\v{q},\tilde{\v{q}})\geq 1-F(\v{q}^{\downarrow},\tilde{\v{q}}^{\downarrow})$, we see that to have the minimal value of error, $\tilde{\v{q}}^{\downarrow}$ needs to maximize the fidelity $F(\v{q}^{\downarrow},\tilde{\v{q}}^{\downarrow})$ subject to the constraint given in Eq.~\eqref{constrainoflagrangemultiplier}. Thus we can write
\begin{equation}\label{eq:Optimal_q}
\begin{split}
 &\max ~~ \qquad\frac{1}{L}\left(\sum_{i=1}^L \sqrt{\tilde{q}_i^{\downarrow}}\right)^2\\
 &\textrm{such that} \quad \sum_{i=1}^L \tilde{q}_i^{\downarrow}=1-\epsilon.
\end{split}
\end{equation}
The Lagrangian of the aforementioned optimization problem in Eq.~\eqref{eq:Optimal_q} is given by 
\begin{equation}
    \mathcal{L}= \frac{1}{L}\left(\sum_{i=1}^L \sqrt{\tilde{q}_i^{\downarrow}}\right)^2-\lambda\left(\sum_{i=1}^L \tilde{q}_i^{\downarrow}-1+\epsilon\right),
\end{equation}
where $\lambda$ is a Lagrange multiplier. To find the solution of the problem we calculate 
\begin{eqnarray}
\frac{\partial \mathcal{L}}{\partial \tilde{q}_j^{\downarrow}}= \frac{1}{\Big(L\sqrt{\tilde{q}_j^{\downarrow}}\Big)}\left(\sum_{i=1}^L\sqrt{\tilde{q}_i^{\downarrow}}\right)-\lambda
\end{eqnarray}
for all $j \in \{1,\ldots,L\}$. Solving $\partial \mathcal{L}/\partial \tilde{q}_j^{\downarrow}=0$ we get 
\begin{eqnarray}\label{eq:qjtilde}
\tilde{q}^{\downarrow}_j=\Bigg(\frac{1}{\lambda L}\left(\sum_{i=1}^L\sqrt{\tilde{q}_i^{\downarrow}}\right)\Bigg)^2
\end{eqnarray}
for all $j \in \{1,\ldots,L\}$. Substituting $\tilde{q}^{\downarrow}_j$ from Eq.~\eqref{eq:qjtilde} to the constraint specified in Eq.~\eqref{constrainoflagrangemultiplier} gives 
\begin{equation}\label{eq:lambda_soln}
    \lambda=\frac{1}{\sqrt{(1-\epsilon)L}}\left(\sum_{i=1}^L\sqrt{\tilde{q}_i^{\downarrow}}\right).
\end{equation}
Therefore, putting $\lambda$ from Eq.~\eqref{eq:lambda_soln} into Eq.~\eqref{eq:qjtilde} finally gives
\begin{eqnarray}\label{eq:ordertilde}
\quad\forall i\in\{1,\ldots, L\}:\qquad \tilde{q}_i^{\downarrow} = \frac{1-\epsilon}{L}.
\end{eqnarray}
Since fidelity is a concave function, and the constraint defined in Eq.~\eqref{eq:Optimal_q} is linear, this solution is optimal.

Next, we note that $\v{p}\succ \tilde{\v{q}}$ implies $Q\v{p}^{\downarrow}=\tilde{\v{q}}^{\downarrow}$ where $Q$ is a bistochastic matrix. Since any $Q$ can be decomposed as a convex sum of permutations, we can write
\begin{equation}\label{Qmatrix}
    Q=\sum_{i}\alpha_i\Pi_i,
\end{equation}
where $\Pi_i$ is a permutation matrix and \mbox{$\sum_i \alpha_i=1$} with $\alpha_i\in[0,1]$ for all~$i$. Let us now define a vector $\v{v}$ as 
\begin{equation}
    \v{v}:=\Big(\underbrace{1,\ldots,1}_{L},0\ldots,0\Big).
\end{equation}
Using this $\v{v}$, we can write the following
\begin{eqnarray}\label{eq:QTv}
     &&\v{v}\cdot Q\v{p}^{\downarrow}= \v{v}\cdot\tilde{\v{q}}^{\downarrow}=\v{v}\cdot\v{p}^{\downarrow}=(1-\epsilon)\nonumber\\
     &\Rightarrow& 
     \sum_{i}\alpha_i\v{v}\cdot\Pi_i\v{p}^{\downarrow}= \v{v}\cdot\v{p}^{\downarrow}
\end{eqnarray}
where we use Eq.~\eqref{constrainoflagrangemultiplier} in the first line, and the convex decomposition of $Q$ from Eq.~\eqref{Qmatrix} in the second line. Note that Eq.~\eqref{eq:QTv} can be equivalently written as 
\begin{eqnarray}\label{eq:Trans_equality}
    \sum_{i}\alpha_i(\Pi_i^T\v{v})\cdot\v{p}^{\downarrow}= \v{v}\cdot\v{p}^{\downarrow},
\end{eqnarray}
where $\Pi_i^T$ is the transpose of the permutation matrix~$\Pi_i$. Because the elements of $\v{v}$ and $\v{p}^{\downarrow}$ are arranged in a decreasing order, the maximum value of $(\Pi_i^T\v{v})\cdot\v{p}^{\downarrow}$ is $\v{v}\cdot\v{p}^{\downarrow}$. Thus, we see that equality in Eq.~\eqref{eq:Trans_equality} holds if and only if $\v{v}$ is invariant under $(\Pi_i^T)$ for all $i$. Therefore, $\Pi_i$ can be only of the form 
\begin{equation}
    \label{eq:PiDecomposition}
    \Pi_i = \Pi_i^L\oplus\Pi_i^{L^c},
\end{equation}
where $\Pi_i^L$ and $\Pi_i^{L^c}$ are permutations that act trivially for indices $i>L$ and $i\in\{1,\ldots,L\}$, respectively. Thus, we infer using Eq.~\eqref{Qmatrix} that $Q$ is a block-diagonal matrix of the form
\begin{equation}
    Q=\tilde{B}\oplus B,
\end{equation}
where $\tilde{B}$ and $B$ are $L\times L$ and $(d-L)\times(d-L)$ bistochastic matrices. Thus, 
\begin{eqnarray}
    \label{eq:ordertilde2}
    \forall i > L ,\quad \tilde{q}_i^{\downarrow} = \sum_{j>L}B_{ij}p_j^{\downarrow}.
\end{eqnarray}
Combining Eq.~\eqref{eq:ordertilde} and Eq.~\eqref{eq:ordertilde2} we conclude that $\tilde{\v{q}}^{\downarrow}={\v{q}}^{*\downarrow}$ with $\v{q}^{*}$ defined in Eq.~\eqref{eq:qstar}. This finally implies $\tilde{\v{q}}=\Pi\v{q}^{*}$ with $\Pi$ being some permutation. Moreover, it is straightforward to show that the minimal entropy of $\v{q}^{*}$ is achieved for $B$ being the identity matrix, since the entropy is increasing under application of any non-trivial bistochastic matrix.
\end{proof}

% ------------------------------------------------
% APPENDIX C
% ------------------------------------------------

\section{Eliminating the logarithmic term}
\label{app:log}

We start with the following lemma that will be needed to prove our claim.

\begin{lem}\label{Orderoflogpro}
    For a fixed value of $b>0$ and any $\v{s}$, such that $\|\v{s}\|=\sqrt{\sum_{i=1}^ds_i^2}\leq b$, we have
    \begin{equation}
        \log P^N_{\v{k}(\v{s})}= O(\log N),
    \end{equation}
    where $P^N_{\v{k}}$ and $\v{k}(\v{s})$ are defined by Eqs.~\eqref{eq:Multinomial_prob_vector}~and~\eqref{eq:k_of_s}, respectively.
\end{lem}
\begin{proof}
We start by using the definition,
\begin{equation}\label{Multinomialwiths}
    P^N_{\v{k}(\v{s})} = \binom{N}{k_1(s_1),...,k_d(s_d)}p^{k_1(s_1)}_1 ... p^{k_d(s_d)}_d,
\end{equation}
to write $\log P^N_{\v{k}(\v{s})}$ as
\begin{equation}
   \!\!\!\log P^N_{\v{k}(\v{s})}\!=\!\log N! -\sum_{i=1}^d \log k_i(s_i)!+\sum_{i=1}^d k_i(s_i)\log p_i.\!
\end{equation}
Employing Stirling inequality, 
\begin{align}
    &\log\sqrt{N}+N\log N-N\nonumber\\
    &\qquad\leq\log N!\leq \nonumber\\
    &\qquad\qquad 1+\log\sqrt{N}+N\log N-N,
    \label{Stirlingineq}
\end{align}
we first provide a lower bound for $\log P^N_{\v{k}(\v{s})}$,
\begin{align}
    \log P^N_{\v{k}(\v{s})}\geq& (\log\sqrt{N}+N\log N-N)+\sum_{i=1}^d k_i(s_i)\log p_i\nonumber\\
    &-\sum_{i=1}^d(1\!+\!\log\sqrt{k_i(s_i)}\!+\!k_i(s_i)\log k_i(s_i)\!-\!k_i(s_i)) \nonumber\\
    =&-\sum_{i=1}^d k_i(s_i)\log\Big(\frac{k_i(s_i)}{Np_i}\Big)-d\nonumber\\
    &+\frac{1}{2}\Big(\log N-\sum_{i=1}^d\log k_i(s_i)\Big).
    \label{eq:UpperBound}
\end{align}
Recall that $k_i(s_i)=Np_i+\sqrt{N}s_i$, which implies $\sum_{i=1}^d s_i=0$. To simplify the above further, we lower bound the first term by employing the inequality \mbox{$\log(1+g)<g$} for $g>-1$ in the following way:
\begin{align}
     &\sum_{i=1}^d k_i(s_i)\log\Big(\frac{k_i(s_i)}{Np_i}\Big)=\sum_{i=1}^d k_i(s_i)\log\left(1+\frac{s_i}{\sqrt{N}p_i}\right)\nonumber\\
     &\qquad\leq\sum_{i=1}^d k_i(s_i)\frac{s_i}{\sqrt{N}p_i}=\sqrt{N}\sum_{i=1}^d(1+\frac{s_i}{\sqrt{N}p_i})s_i\nonumber\\
     &\qquad\qquad=\sum_{i=1}^d \frac{s_i^2}{p_i}\leq\sum_{i=1}^d \frac{s_i^2}{p_{\min}}\leq \frac{b^2}{p_{\min}},
\end{align}
where $p_{\min}=\text{min }\{p_1,\ldots, p_d\}$. Moreover, observing that
\begin{equation}
     \log N \geq\left(\log N-\sum_{i=1}^d\log k_i(s_i)\right) \geq -(d-1)\log N ,
\end{equation}
we can conclude that 
\begin{equation}\label{eq:O(LogN)}
    \log N-\sum_{i=1}^d\log k_i(s_i)=O(\log N).
\end{equation}
Putting it together, we further simplify the bound given in Eq.~\eqref{eq:UpperBound} as
\begin{equation}
     \!\!\!\log P^N_{\v{k}(\v{s})}\geq -d-\frac{b^2}{p_{\min}}-\frac{(d-1)}{2}\log{N}=O(\log N).
\end{equation}

Similarly, by employing the Stirling inequality from Eq.~\eqref{Stirlingineq}, we also prove an upper bound for $\log P^N_{\v{k}(\v{s})}$ as follows
\begin{align}
    \log P^N_{\v{k}(\v{s})}
    \leq& (1+\log{\sqrt{N}}+N\log N-N)+\sum_{i=1}^dk_i(s_i)\log p_i\nonumber\\
    &-\sum_{i=1}^d\Big(\log\sqrt{k_i(s_i)}+k_i(s
    _i)\log k_i(s_i)-k_i(s_i)\Big)\nonumber\\
    =&-\sum_{i=1}^dk_i(s_i)\log\Big(\frac{k_i(s_i)}{Np_i}\Big)+1\nonumber\\
    &+\frac{1}{2}(\log N-\sum_{i=1}^d\log k_i(s_i)).
\end{align}
Using the inequality $\log(1+g)\geq\frac{g}{1+g}$ for $g>-1$, we have
\begin{align}
    &\sum_{i=1}^dk_i(s_i)\log\Big(\frac{k_i(s_i)}{Np_i}\Big)= \sum_{i=1}^dk_i(s_i)\log\Big(1+\frac{s_i}{\sqrt{N}p_i}\Big)\nonumber\\
    &\qquad\geq\sum_{i=1}^dk_i(s_i)\frac{s_i}{s_i+\sqrt{N}p_i}=\sqrt{N}\sum_{i=1}^d s_i = 0.
    \label{log_kisi_Npi}
\end{align}
The above inequality together with Eq.~\eqref{eq:O(LogN)} imply that $\log P^N_{\v{k}(\v{s})}\leq O(\log N)$ which completes the proof.

\end{proof}

Using Lemma~\ref{Orderoflogpro}, we will now be able to prove our claim that is captured by the following result.
\begin{lem}
The following limits are equal:
\label{Lemma_to_remove_lnN}
\begin{align}
\!\!\!\!\!\!&\lim_{N\rightarrow\infty}\sum_{\v{s}}\Big\{P^N_{\v{k}(\v{s})}\Big| \frac{1}{N}\log(P^N_{\v{k}(\v{s})})+\frac{\beta}{\sqrt{N}}\sum_{i}s_iE_i\nonumber\\
\!\!\!\!\!\!&\phantom{\lim_{N\rightarrow\infty}\sum_{\v{s}}\Big\{P^N_{\v{k}(\v{s})}\Big|}+F_{\rm diss}^N\geq 0\Big\}\nonumber\\
\!\!\!\!\!\! &~~=\lim_{N\rightarrow\infty}\sum_{\v{s}}\Big\{P^N_{\v{k}(\v{s})}\Big| \frac{\beta}{\sqrt{N}}\sum_{i}s_iE_i+F_{\rm diss}^N\geq 0\Big\}.
\end{align}
\end{lem}
\begin{proof}
We start by introducing the following notation
\begin{align}
 &A(N) := \sum_{\v{s}}\Big\{P^N_{\v{k}(\v{s})}\Big| \frac{1}{N}\log(P^N_{\v{k}(\v{s})})+\frac{\beta}{\sqrt{N}}\sum_{i}s_iE_i\nonumber\\
 &\qquad\qquad\qquad\qquad\quad+F_{\text{diss}}^N\geq 0\Big\},\nonumber\\
 &A(b,N) := \sum_{\v{s}}\Big\{P^N_{\v{k}(\v{s})}\Big| \frac{1}{N}\log(P^N_{\v{k}(\v{s})})+\frac{\beta}{\sqrt{N}}\sum_{i}s_iE_i\nonumber\\
 &\qquad\qquad\qquad\qquad\quad+F_{\text{diss}}^N\geq 0\text{  such that  } \|\v{s}\|\leq b\Big\},\nonumber\\
 &B(N) := \sum_{\v{s}}\Big\{P^N_{\v{k}(\v{s})}\Big| \frac{\beta}{\sqrt{N}}\sum_{i}s_iE_i+F_{\text{diss}}^N\geq 0\Big\},\nonumber\\
 &B(b,N) := \sum_{\v{s}}\Big\{P^N_{\v{k}(\v{s})}\Big| \frac{\beta}{\sqrt{N}}\sum_{i}s_iE_i+F_{\text{diss}}^N\geq 0\nonumber\\
 &\qquad\qquad\qquad\qquad\quad\text{  such that  } \|\v{s}\|\leq b\Big\},\nonumber\\
 &\Omega(b,N):= \sum_{\v{s}}\Big\{P^N_{\v{k}(\v{s})}\Big| \text{  such that  } \|\v{s}\|\geq b\Big\}.
 \end{align}
 Our goal is to show that
 \begin{equation}
     \lim_{N\rightarrow\infty} A(N)=\lim_{N\rightarrow\infty} B(N).
 \end{equation}
From the definition it follows that 
\begin{subequations}
\begin{eqnarray}
     A(N)-\Omega(b,N)&\leq& A(b,N)\leq A(N)\label{LimA},\\
      B(N)-\Omega(b,N)&\leq& B(b,N)\leq B(N)\label{LimB}.
\end{eqnarray}
\end{subequations}
Taking the limit $N\rightarrow\infty$ of Eqs.~\eqref{LimA}~and~\eqref{LimB}, we have 
\begin{subequations}
\begin{align}
    \lim_{N\rightarrow\infty}\Big(A(N)-\Omega(b,N)\Big)&\leq\lim_{N\rightarrow\infty} A(b,N)\leq \lim_{N\rightarrow\infty} A(N),\label{LimA1}\\
    \lim_{N\rightarrow\infty}\Big(B(N)-\Omega(b,N)\Big)&\leq \lim_{N\rightarrow\infty}B(b,N)\leq \lim_{N\rightarrow\infty}B(N)\label{LimB1}.
\end{align}
\end{subequations}
Now, let us define
\begin{eqnarray}
     \lim_{N\rightarrow\infty}\Omega(b,N)=:\epsilon(b).
\end{eqnarray}
As the multinomial distribution concentrates around mean for $N\rightarrow\infty$, it follows that $\lim_{b\rightarrow\infty}\epsilon(b)=0$. Therefore, taking the limit $b\rightarrow\infty$ in Eq.~\eqref{LimA1}  we have
\begin{align}
     & \lim_{N\rightarrow\infty}A(N)\leq\lim_{b\rightarrow\infty} \lim_{N\rightarrow\infty}A(b,N)\leq \lim_{N\rightarrow\infty}A(N)\nonumber\\
     &\qquad\Rightarrow\quad\lim_{b\rightarrow\infty} \lim_{N\rightarrow\infty}A(b,N)= \lim_{N\rightarrow\infty}A(N).\label{eq:ANABN}
\end{align}
Analogously, taking the limit $b\rightarrow\infty$ in Eq.~\eqref{LimB1} we can show that
\begin{equation}\label{eq:BNBBN}
    \lim_{b\rightarrow\infty} \lim_{N\rightarrow\infty}B(b,N)=\lim_{N\rightarrow\infty}B(N).
\end{equation}

Moreover, for any fixed $b$, by employing Lemma~\ref{Orderoflogpro}, we see that $\frac{1}{N}\log(P^N_{\v{k}(\v{s})})=O(\frac{\log N}{N})$, which vanishes as \mbox{$N\rightarrow\infty$}. Therefore, we have 
 \begin{equation}
     \lim_{N\rightarrow\infty} A(b,N) =\lim_{N\rightarrow\infty} B(b,N),
 \end{equation}
and so by taking the limit $b\rightarrow\infty$, we arrive at
\begin{eqnarray}\label{eq:ANABNBNBBN}
     \lim_{b\rightarrow\infty}\lim_{N\rightarrow\infty} A(b,N) =\lim_{b\rightarrow\infty}\lim_{N\rightarrow\infty} B(b,N).
\end{eqnarray}
Combining the above with Eqs.~\eqref{eq:ANABN}-\eqref{eq:BNBBN} we have
\begin{equation}
    \lim_{N\rightarrow\infty} A(N) =\lim_{N\rightarrow\infty} B(N),
\end{equation}
which completes the proof.
\end{proof}

% ------------------------------------------------
% APPENDIX D
% ------------------------------------------------

\section{Central limit theorem for multinomial distribution}
\label{app:clt}

\begin{lem}\label{CLT_Multinomial}
    The multinomial distribution with mean \mbox{$\v{\mu}=N\v{p}$} and a covariance matrix $\v{\Sigma}$ can be approximated in the asymptotic limit by a multivariate normal distribution $\mathcal{N}^{(\v{\mu},\v{\Sigma})}$ with mean $\v{\mu}$ and a covariance matrix~$\v{\Sigma}$.
\end{lem}
\begin{proof}
Assume $\v{X_1}\ldots \v{X_N}$ are independent and identically distributed random vectors each
of them with the following distribution
\begin{equation}
\!\!\!\text{Prob} (\v{X}=\v{x}) = 
     \begin{cases}
       \prod_{i=1}^d p_i^{x_i} &\quad\!\text{if $\v{x}$ is unit vector},\\
       0 &\quad\!\text{otherwise}.
     \end{cases}
\end{equation}
Then, the mean vector of $\v{X}$ is $\v{p}$ and the covariance matrix $\frac{1}{N}\v{\Sigma}=\text{ diag }(\v{p})-\v{p}\v{p}^T$. Define  $\v{S}_N:=\v{X}_1+\ldots+\v{X}_N$. Then 
\begin{eqnarray}
     \text{Prob}(\v{S}_N=\v{k}) &=&\binom{N}{k_1\ldots k_d}p_1^{k_1}\ldots p_d^{k_d}.
\end{eqnarray}
We thus see that a multinomial distribution arises from a sum of independent and identically distributed random variables. Therefore, using the central limit theorem, we obtain that the distribution of $\v{k}$ approaches the distribution $\mathcal{N}^{(\v{\mu},\v{\Sigma})}$ arbitrarily well for \mbox{$N\to\infty$}, which completes the proof.
\end{proof}

% ------------------------------------------------
% APPENDIX E 
% ------------------------------------------------

\section{Entropy difference between uniformised and non-uniformised embedding boxes}
\label{app:DDF}

In this appendix, we show that there exists a permutation which transforms the total final state $\hat{\v{F}}^N$ from Eq.~\eqref{eq:distribution-total} into a state that is uniform in almost all embedding boxes, leading to no dissipation up to higher-order asymptotics. We start with the following observations about the total initial and final states. First, since we consider the initial state composed of $N$ independent systems in identical incoherent states $\v{P}^N=\v{p}^{\otimes N}$, we note that $\hat {\v P}^{N} \otimes \hat{\tilde{\v{G}}}^N$ has a polynomial number of distinct entries with exponential degeneracy and exponentially many embedding boxes. Second, note that according to Eq.~\eqref{eq:distribution-total} the entries of the embedded total final state, $\hat{\v{F}}^N$, are essentially given by the entries of $\hat {\v P}^{N} \otimes \hat{\tilde{\v{G}}}^N$ (plus many equal entries $(1-\epsilon)/K$). Thus, $\hat{\v{F}}^N$ has polynomially many distinct entries with exponential degeneracy and exponentially many embedding boxes. Employing permutational freedom, we can then rearrange the entries of $\hat{\v{F}}^N$ so that they are uniform in almost every embedding box, except poly($N$) of them. We will denote the probability distribution over the embedding boxes by $\v{q}$ and note that it is essentially equal to $\v P^{N} \otimes \, \tilde{\v G}^N$. Moreover, we will denote by $\v{r}^{(i)}$ the normalised distribution within the $i$-th embedding box.

The next step is to write the entropy of the embedded total final state as the entropy of probability distribution over different embedding boxes, $H(\v q)$, plus the average entropy of normalised probabilities within each box, $H(\v{r}^{(i)})$,
\begin{equation}
\label{eq:Hnormalisedbox}    
    H(\hat{\v{F}}^N) = H(\v{q}) + \sum_i q_i H(\v{r}^{(i)}).
\end{equation}
Note that whether we uniformise or not a given box, the entropy $H(\v q)$ does not change. Consequently, the entropy of the final state uniformised within each embedding box takes the form of
\begin{align}
\label{eq:entr-totalfinuni}
    H(\hat{\v{F}}_{\rm{uni}}^N) &= H(\v{q}) + \sum_i q_i H(\v{r}^{(i)}_{\rm{uni}}) ,
\end{align}
with $\v{r}^{(i)}_{\rm{uni}}$ representing the normalised and uniformised probability within the $i$-th embedding box. Thus, the entropy difference between the uniformised and non-uniformised distributions reads 
\begin{equation}
    H(\hat{\v{F}}_{\rm{uni}}^N) - H(\hat{\v{F}}^N) = \sum_i q_i \left ( H(\v{r}^{(i)}_{\rm{uni}})-H(\v{r}^{(i)}) \right) .
\end{equation}

Now, note that due to the previous argument, the above sum is performed only over a polynomial number of boxes. Denoting the set with size poly($N$) as $\Omega$, one can write
\begin{align}
  H(\hat{\v{F}}_{\rm{uni}}^N) - H(\hat{\v{F}}^N) &= \sum_{i\in\Omega} q_i \left ( H(\v{r}^{(i)}_{\rm{uni}})-H(\v{r}^{(i)}) \right) \nonumber\\ &\leq \sum_{i\in\Omega} q_iH(\v{r}^{(i)}_{\rm{uni}}).
\end{align}
Let us analyse the right-hand side of the above equation. First, note that $q_i$ is exponentially small, $q_i \propto \exp(-N)$; while $H(\v{r}^{(i)}_{\rm{uni}})$ is the entropy of a uniform state over the dimension of the $i$-th embedding box (equal to $\prod_i D^{k_i}_i$), so it will scale linearly in $N$:
\begin{equation}
    H(\v{r}^{(i)}_{\rm{uni}})=\log \left( \prod_i D^{k_i}_i\right) \propto O(N) \, .
\end{equation}
Therefore, we conclude that the entropy difference vanishes exponentially in $N$
\begin{equation}
H(\hat{\v{F}}_{\rm{uni}}^N) - H(\hat{\v{F}}^N) \propto \exp(-N).
\end{equation}

% ------------------------------------------------
% BIBLIOGRAPHY
% ------------------------------------------------

\bibliography{bibliography}

\end{document}